\documentclass{article}
\usepackage{lipsum}
\usepackage{tikz}
\usepackage{xifthen}
\usepackage{listings}
\RequirePackage{amsthm,amsmath,amssymb}
\usepackage{natbib}
\usepackage{authblk}
\usepackage{float}

\newcounter{num}
\newcommand{\plot}[4]
{      \includegraphics[width=#1\textwidth,bb=0 0 #2 #3]{23-\arabic{#4}.pdf}
\addtocounter{#4}{1}
}

\makeatletter
\def\@maketitle{%
  \newpage
  \null
  \vskip 2em%
  \begin{center}%
  \let \footnote \thanks
    {\Large\bfseries \@title \par}%
    \vskip 1.5em%
    {\normalsize
      \lineskip .5em%
      \begin{tabular}[t]{c}%
        \@author
      \end{tabular}\par}%
    \vskip 1em%
    {\normalsize \@date}%
  \end{center}%
  \par
  \vskip 1.5em}
\makeatother

\newcommand{\doBlank}[1]{}

\def\?#1{}

\newcommand{\dif}{\mathrm{d}}

\numberwithin{equation}{section}
\theoremstyle{plain}

\newtheorem{assumption}{Assumption}
\newtheorem{lemma}{Lemma}[section]
\newtheorem{definition}{Definition}
\newtheorem{proposition}{Proposition}[section]

\newtheorem{theorem}{Theorem}[section]

\newcommand{\writetitle}{0}
\newcommand{\mytitle}[1]
{   \ifthenelse{\writetitle=1}{}{}
}

\newread\mysource
\begin{document}
\title{Efficient strategy for the Markov chain Monte Carlo  in high-dimension with heavy-tailed target probability distribution}
\author{Kengo KAMATANI%
\thanks{Supported in part by Grant-in-Aid for Young Scientists (B) 24740062.}}
\affil{Graduate School of Engineering Science, Osaka University and JST, CREST}

\date{Dated: \today}

\maketitle
\begin{abstract} 
The purpose of this paper is to introduce a new Markov chain Monte Carlo method and exhibit its efficiency by 
simulation and  high-dimensional asymptotic theory. 
Key fact is that our algorithm has a reversible proposal transition kernel, which is designed to have a heavy-tailed invariant probability distribution. 
The high-dimensional asymptotic theory is studied for a class of heavy-tailed target probability distribution. 
As the number of dimension of the state space  goes to infinity, we will show that our algorithm has a much better convergence rate than that of the preconditioned Crank Nicolson (pCN) algorithm
and the  random-walk Metropolis (RWM) algorithm. 
We also show that our algorithm is at least as good as the pCN algorithm 
and better than the RWM algorithm
for light-tailed target probability distribution. 
\end{abstract}

\section{Introduction}
The Markov chain Monte Calro (MCMC) method is a widely used technique for evaluation of complicated integrals, 
especially in high dimensional setting. A lot of new methods are developed in the past few decades. 
However it is still very difficult to 
choose an MCMC that works well for a given function and a given measure, which is called the target (probability) distributoin. The choice of MCMC heavily depends on the tail behaviour of the target probability distribution. In particular, it is well-known that many MCMC algorithms 
behave poorly for heavy-tailed target probability distribution. 

In our previous work, in \citet{arXiv:1406.5392}, 
we studied some asymptotic properties of the random-walk Metropolis (RWM) algorithm for heavy-tailed target probability distribution. 
To perform RWM algorithm,  we have to choose a proposal distribution. This choice heavily affects the performance. 
We showed that the most standard choice, the Gaussian proposal distribution attains the optimal rate of convergence, although  this rate is quite poor. This rather disappointing fact illustrates that the RWM algorithm can not be so good. 
To find a more efficient strategy is an important unsolved problem. 

A candidate of this, the preconditioned Crank-Nicolson (pCN) algorithm is first appeared in \citet{MR2537193}. 
The method is a simple modification of a classical Gaussian RWM algorithm and so their computational costs are almost the same. 
The efficiency for this simple candidate  was provided  in simulation by \citet{MR3135540}
and its theoretical benefit was provided in \citet{MR2537193}, \citet{arXiv:1108.1494v2}, \citet{MR3161650} and \citet{MR3262508}. 
However our simulation shows that it works well only for a specific light-tailed target distribution and works quite poor otherwise, in particular, for heavy-tailed target probability distribution
(in Theorem \ref{Theorem-1}, we will prove it in terms of the convergence rate). 

In this paper, we introduce a new algorithm which is a slight modification of the original pCN algorithm though their performances are completely different.  It works well and is quite robust. 
Let us describe our new algorithm, 
the mixed preconditioned Crank-Nicolson (MpCN) algorithm. Let $P(\dif x)=p(x)\dif x$ be the target probability distribution on $\mathbb{R}^d$. 
Fix $\rho\in (0,1)$. Set initial value $x=(x_1,\ldots, x_d)\in\mathbb{R}^d$ and
let $\|x\|=(\sum_{i=1}^d x_i^2)^{1/2}$. 
The algorithm goes as follows:

\begin{itemize}
\item Generate $r\sim \mathrm{Gamma}(d/2,\|x\|^2/2)$. 
\item Generate $x^*= \rho^{1/2} x+(1-\rho)^{1/2} r^{-1/2}w$ where $w$ follows the standard normal distribution. 
\item Accept $x^*$ as $x$ with probability $\alpha(x,x^*)$, and otherwise,  discard $x^*$, where
\begin{equation}\nonumber
\alpha(x,y)=\min\left\{1, \frac{p(y)\|x\|^{-d}}{p(x)\|y\|^{-d}}\right\}. 
\end{equation}
\end{itemize}
In the above, $\mathrm{Gamma}(\nu,\alpha)$ is the Gamma distribution with the shape parameter $\nu$
and the scale parameter $\alpha$ with the probability distribution function $\propto x^{\nu-1}\exp(-\alpha x)$. 
In our simulation, we set $\rho=0.8$. Key fact is that the proposal transition kernel of the algorithm 
has a heavy-tailed invariant probability distribution. 
Thus it is not surprising if the new method works better than the pCN algorithm for heavy-tailed target probability distribution. 
However we will show that the new method has the same convergence rate as the pCN algorithm even for light-tailed target probability distribution. Our method is robust, which is one of the most important property for MCMC. 

We study its theoretical properties via high-dimensional asymptotic theory. 
 The high-dimensional asymptotic theory for MCMC was first appeared in \citet{MR1428751} and further developed in \citet{RSSB:RSSB123}. See \citet{MR3135540} for recent results.  
We use this framework together with the study of consistency of MCMC by \citet{Kamatani10}.  

The main technical tools are  Malliavin calculus and Stein's techniques. 
The reader is referred to \citet{MR2200233} for the former and \citet{Stein}
for the latter and see \citet{MR3003367} for the connection of the two fields. 
The analysis of this connection is a very active area of research and our paper illustrates usefulness of the analysis
even for Bayesian computation.

The paper is organized as follows. The numerical simulations are provided in the right after this section. 
We also illustrate the limitation of the MpCN algorithm in Section \ref{sim4}. 
In Section \ref{highdimsection}, high-dimensional asymptotic properties will be studied.  We will show that the pCN algorithm is worse than the classical RWM algorithm 
for heavy-tailed target probability distribution. On the other hand, the MpCN algorithm attains a better rate than the RWM algorithm. 
Proofs are relegated to Section \ref{proofsec}. 
In the appendix, Section \ref{appen1} includes a short introduction to Malliavin calculus
and Stein's techniques. 
Section \ref{appen2} provides  some properties for consistency of MCMC.

Finally, we note that our new algorithm was already implemented for the Bayesian type estimation for ergodic diffusion process in  \citet{UK} (More precisely, a version of MpCN. See Section \ref{discuss} for the detail). 
The target probability distribution is  very complicated although it is not heavy-tailed.  The performance of the Gaussian RWM algorithm 
was quite poor due to the complexity. However the new method worked well as described in Figure 1 of \citet{UK}.
In our current study, we only describe usefulness of our algorithm for a class of heavy-tailed target probability distribution. 
However, this heavy-tail assumption is just an example of target probability distribution that is difficult to approximate by MCMC. Our method is robust, and we believe that the method is useful for non heavy-tailed  complicated target probability distribution as illustrated in  \citet{UK}.

\subsection{Notation}
Several norms are considered in this paper. 
\begin{itemize}
\item
For $x=(x_1,\ldots, x_k),\ y=(y_1,\ldots, y_k)\in\mathbb{R}^k$, write $\|x\|=\left(\sum_{i=1}^k x_i^2\right)^{1/2}$
and $\left\langle x,y\right\rangle=\sum_{i=1}^k x_iy_i$. 
If $h$ is in a Hilbert space $\mathfrak{H}$ with inner product $\left\langle \cdot,\cdot\right\rangle_{\mathfrak{H}}$, write $\|h\|_{\mathfrak{H}}=\left(\left\langle h,h\right\rangle_{\mathfrak{H}}\right)^{1/2}$.  
\item
For a function $f:E\rightarrow\mathbb{R}$, write $\|f\|_\infty=\sup_{x\in E}|f(x)|$. 
\item
If $F$ is a real valued random variable on an abstract Wiener space $(W,\mathfrak{H},\mathbb{P})$, 
write 
$\|F\|_{\mathbb{D}^{1,2}}=(\mathbb{E}[F^2]+\mathbb{E}[\|DF\|^2_\mathfrak{H}])^{1/2}$. 
When the abstract Wiener space $(W,\mathfrak{H},\mathbb{P}_y)$  depends on $y\in(0,\infty)$, 
write 
$\|F\|_{\mathbb{D}^{1,2}_\delta}=\sup_{y\in [\delta,\delta^{-1}]}(\mathbb{E}_y[F^2]+\mathbb{E}_y[\|DF\|^2_\mathfrak{H}])^{1/2}$
for $\delta\in (0,1)$. 
\item If $\nu$ is a signed measure on $(E,\mathcal{E})$, write $\|\nu\|_{\mathrm{TV}}=\sup_{A\in \mathcal{E}}|\nu(A)|$. 
The integral with respect to $\nu$ is denoted by $\nu(f)=\int_E f(x)\nu(\dif x)$. 
In particular, $Nf =\mathbb{E}[f(X)],\ X\sim N(0,1)$. 
\end{itemize} 
Write $N_d(\mu,\Sigma)$ for the $d$-dimensional normal distribution
with  mean  $\mu\in\mathbb{R}^d$ and  variance covariance matrix $\Sigma$, 
and $\phi_d(x;\mu,\Sigma)$ be its probability distribution function. 
When $d=1$, write $N(\mu,\sigma^2)$ and $\phi(x;\mu,\sigma^2)$ with respectively. 
We also denote  the $d$-dimensional standard normal distribution briefly by $N_d$
and write $N=N_1$. 
Write $I_d$ for the  $d\times d$-identity matrix.
Write $\mathcal{L}(X)$ or $\mathcal{L}_\mathbb{P}(X)$ for the law of random variable $X$. 
Write $X_n\Rightarrow X$ if the law of $X_n$ converges weakly to that of $X$. 
Write $X|Y$ for the conditional distribution of $X$ given $Y$.

\section{The MpCN algorithm and its performance}\label{perform}

In this section, we describe two Metropolis-Hastings algorithms. 
The Metropolis-Hastings algorithm generates a Markov chain $\{X_m\}_m$ with transition kernel $K(x,\dif y)$ on $(E,\mathcal{E})$ defined by
the following: Set $X_0\in E$  and for $m\ge 1$, 
\begin{equation}\nonumber
\left\{\begin{array}{l}
X^*_m\sim R(X_{m-1},\dif x)\\
X_m=\left\{\begin{array}{ll}
X^*_m & \mathrm{with\ probability}\ \alpha(X_{m-1},X_m^*)\\
       X_{m-1} & \mathrm{with\ probability}\ 1-\alpha(X_{m-1},X_m^*)
       \end{array}\right.
\end{array}\right. 
\end{equation}
where $R(x,\dif y)$ is called the proposal transition kernel, and $\alpha(x,y)$ is called the acceptance ratio
that satisfy
\begin{equation}\label{detailedbalance}
P(\dif x)R(x,\dif y)\alpha(x,y)=P(\dif y)R(y,\dif x)\alpha(y,x)
\end{equation}
where $P(\dif x)$ is the target probability distribution. 
The Markov chain is called reversible with respect to $P(\dif x)$ if 
\begin{equation}\nonumber
P(\dif x)K(x,\dif y)=
P(\dif y)K(y,\dif x). 
\end{equation}
If the acceptance ratio satisfies (\ref{detailedbalance}), then the Markov chain has reversibility. 
See monograph \citet{RC} or review \citet{TierneyAOS94} for further details. 

\subsection{The pCN algorithm}

Let $P_d$ be a probability measure on $\mathbb{R}^d$ with density $p_d(x)$. 
In this paper, 
the following algorithm that generate a Markov chain $X^d=\left\{X^d_m\right\}_{m\in\mathbb{N}_0}$ is called the preconditioned Crank-Nicolson (pCN) algorithm  
for the target probability distribution $P_d$ if $X^d_0$ is a $\mathbb{R}^d$-valued random variable, and 
for $m\ge 1$, 
\begin{equation}\label{pCN}
\left\{\begin{array}{l}
X^{d*}_m= \sqrt{\rho}X^d_{m-1}+\sqrt{1-\rho}W^d_m,\ W^d_m\sim N_d(0,I_d)\\
X^d_m=\left\{\begin{array}{ll}
X^{d*}_m & \mathrm{with\ probability}\ \alpha_d(X^d_{m-1},X^{d*}_m)\\
       X^d_{m-1} & \mathrm{with\ probability}\ 1-\alpha_d(X^d_{m-1},X^{d*}_m)
       \end{array}\right.
\end{array}\right. 
\end{equation}
where $\alpha_d(x,y)=\min\left\{1, p_d(y)\phi_d(x;0,I_d)/p_d(x)\phi_d(y;0,I_d)\right\}$. 
Write $\mathrm{pCN}(P_d)$ for the law $\mathcal{L}(X^d)$  if $X_0^d\sim P_d$. 
The conditional distribution $X_m^{d*}|X_{m-1}^d$ is given by the following joint distribution:
\begin{equation}\nonumber
(X_{m-1}^d, X_m^{d*})\sim N_{2d}\left(0, \left(\begin{matrix}I_d&\sqrt{\rho}I_d\\\sqrt{\rho}I_d&I_d\end{matrix}\right)\right).
\end{equation}
In particular, if $P_d=N_d(0,I_d)$ and $X_0^d\sim P_d$, each $X_m^d$ is always accepted
and 
it becomes a $d$-dimensional $\mathrm{AR}(1)$ process. 

\subsection{The MpCN algorithm}

In this paper, we propose 
the following algorithm that generate a Markov chain $X^d=\left\{X^d_m\right\}_{m\in\mathbb{N}_0}$:
 Set $X^d_0$ as a $\mathbb{R}^d$-valued random variable, and for $m\ge 1$, 
\begin{equation}\label{MpCN}
\left\{\begin{array}{l}
Z^d_m\sim \mathrm{InvGamma}(d/2,\|X_{m-1}^d\|^2/2)\\
X^{d*}_m= \sqrt{\rho}X^d_{m-1}+\sqrt{(1-\rho)Z_m^d}W^d_m,\ W^d_m\sim N_d(0,I_d)\\
X^d_m=\left\{\begin{array}{ll}
X^{d*}_m & \mathrm{with\ probability}\ \alpha_d(X^d_{m-1},X^{d*}_m)\\
       X^d_{m-1} & \mathrm{with\ probability}\ 1-\alpha_d(X^d_{m-1},X^{d*}_m)
       \end{array}\right.
\end{array}\right. 
\end{equation}
where $\alpha_d(x,y)=\min\left\{1, p_d(y)\|x\|^{-d}/p_d(x)\|y\|^{-d}\right\}$, 
and $\mathrm{InvGamma}(\nu,\alpha)$ is the inverse Gamma distribution with the shape parameter 
$\nu$ and the scale parameter $\alpha$ with density
\begin{equation}\nonumber
g(z;\nu,\alpha)=1_{(0,\infty)}(z)\frac{\alpha^\nu}{\Gamma(\nu)}z^{-\nu-1}\exp(-\alpha/z). 
\end{equation}
In this paper, this algorithm is called the mixed preconditioned Crank-Nicolson (MpCN) algorithm
for the target probability distribution $P_d$. 
Write $\mathrm{MpCN}(P_d)$  for the law $\mathcal{L}(X^d)$ if $X_0^d\sim P_d$. 
Formally, the conditional distribution $X_m^{d*}|X_{m-1}^d$ is given by the following joint distribution: 
\begin{equation}\nonumber
(X_{m-1}^d, X_m^{d*})|Z_m^d\sim N_{2d}\left(0, Z_m^d\left(\begin{matrix}I_d&\sqrt{\rho}I_d\\\sqrt{\rho}I_d&I_d\end{matrix}\right)\right), Z_m^d\sim \overline{Q}
\end{equation}
when $\overline{Q}(\dif x)=1_{(0,\infty)}(x)x^{-1}\dif x$. 
By this structure, the transition kernel is reversible with respect to 
\begin{equation}\label{stationprop}
\overline{P}_d(\dif x)\propto \int_{z\in (0,\infty)}\phi_d(x,0,zI_d)\overline{Q}(\dif z)\dif x\propto \|x\|^{-d}\dif x. 
\end{equation}
Since $\overline{P}_d$ and $\overline{Q}$ are  improper (not probability measures but $\sigma$-finite measures), the above argument is just a formal sense. This argument is justified by the following. 

\begin{lemma}\label{reverse}
The proposal transition kernel of the MpCN algorithm is reversible with respect to a $\sigma$-finite measure $\overline{P}_d(\dif x)=\overline{p}_d(x)\dif x=\|x\|^{-d}\dif x$, and the transition kernel of the MpCN algorithm 
is reversible with respect to $P_d$. 
\end{lemma}

\begin{proof}
Write $R_d(x,\dif x^*)=\int_{z\in(0,\infty)}R_d(x,\dif z)R_d(x,z,\dif x^*)$ for the proposal transition kernel of the MpCN algorithm where
\begin{equation}\nonumber
R_d(x,\dif z)= \mathrm{InvGamma}(d/2,\|x\|^2/2)\\,\ R_d(x,z,dx^*)=N_d(\sqrt{\rho}x,(1-\rho)zI_d).
\end{equation}
Then 
\begin{equation}\nonumber
\overline{P}_d(\dif x)R_d(x,\dif z)R_d(x,z,\dif x^*)\propto 
\phi_{2d}\left(\left(\begin{array}{c}x\\x^*\end{array}\right);\left(\begin{matrix}0\\0\end{matrix}\right), z\left(\begin{matrix}I_d&\sqrt{\rho}I_d\\\sqrt{\rho}I_d&I_d\end{matrix}\right)\right)\frac{\dif z}{z}\dif x\dif x^*. 
\end{equation}
Since the right-hand side is exchangeable with respect to $x$ and $x^*$,  the proposal transition kernel $R_d(x,\dif x^*)$ is reversible with respect to $\overline{P}_d$. 
For the latter case, it is sufficient to show
\begin{equation}\nonumber
P_d(\dif x)R_d(x,\dif x^*)\alpha_d(x,x^*)
=
P_d(\dif x^*)R_d(x^*,\dif x)\alpha_d(x^*,x). 
\end{equation}
However, the left-hand side of the above is
\begin{equation}\nonumber
\overline{P}_d(\dif x)R_d(x,\dif x^*)\frac{p_d}{\overline{p}_d}(x)\alpha_d(x,x^*)
=
\overline{P}_d(\dif x)R_d(x,\dif x^*)\min\left\{\frac{p_d}{\overline{p}_d}(x),\frac{p_d}{\overline{p}_d}(x^*)\right\}. 
\end{equation}
Since $R_d(x,dx^*)$ is reversible with respect to $\overline{P}_d$, 
the right-hand side of the above is again, exchangeable with respect to $x$ and $x^*$. 
Hence the claim follows. 
\end{proof}

\subsection{Numerical results}
We consider two kinds of numerical experiments. \\

\noindent
\textbf{Efficiency of MpCN algorithm:}
In Sections \ref{sim1}-\ref{sim3}, we illustrate efficiency of the MpCN algorithm. We will compare two RWM algorithms and the pCN and MpCN algorithms with $M=10^8$ iterations (no burn-in) for each. 
The algorithms we consider are
\begin{enumerate}
\item The RWM algorithm with Gaussian proposal distribution. 
More precisely, the update $x^*$ from the current value $x$ is generated by 
$x^*=x+\sigma_d\epsilon$ where $\epsilon$ follows the standard normal distribution and 
$\sigma_d^2=1/d$ in this simulation.  
\item The RWM algorithm with the $t$-distribution as the proposal distribution (two degrees of freedom). 
More precisely, $x^*=x+\sigma_d\epsilon$ where $\epsilon$ follows the $t$-distribution with two degrees of freedom 
and $\sigma_d^2=1/d$ in this simulation.  
\item The pCN algorithm for $\rho=0.8$. 
\item The MpCN algorithm for $\rho=0.8$. 
\end{enumerate}
The target probability distributions are the following. 
\begin{enumerate}
\item[(a)] The standard normal distribution. 
\item[(b)] The $t$-distribution (two degrees of freedom). 
\item[(c)] A perturbation of the $t$-distribution. 
\end{enumerate}

For each target probability distribution and each algorithm, we generate a single 
Markov chain $\left\{X_m^d\right\}_m$ 
with initial value $X_0^d\sim N_d(0,I_d)$ and plot four figures
as in Figure \ref{Fig1}. 
\begin{figure}[H]
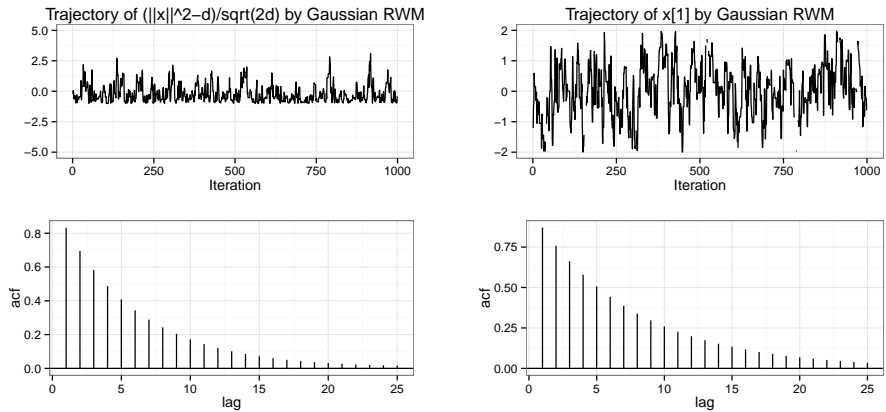

\centering
\plot{0.35}{360}{360}{num}
\plot{0.35}{360}{360}{num}
\caption{The RWM algorithm with Gaussian proposal distribution for $P_d=N_d(0,I_d)$ for $d=2$.\label{Fig1}} 
\end{figure}
This example is just for an illustration. The target probability distribution is the two dimensional standard normal distribution and the MCMC is  the RWM algorithm with Gaussian proposal distribution. These four plots are
\begin{enumerate}
\item[(i)] Trajectory of the normalised distance from the origin. 
When the target probability distribution is  the standard normal distribution, we plot $\left\{(2d)^{-1/2}(\|X_m^d\|^2-d)\right\}_m$
and for other cases, we plot $\left\{\|X_m^d\|^2/d\right\}_m$ (upper left). 
\item[(ii)] The autocorrelation plot of the above (bottom left). 
\item[(iii)] Trajectory $\{X_{m,1}^d\}_m$ where $X_m^d=(X_{m,1}^d,\ldots, X_{m,d}^d)$ (upper right). 
\item[(iv)] The autocorrelation plot of the above (bottom right). 
\end{enumerate}
The simulation results are illustrated in Sections \ref{sim1}-\ref{sim3}. \\

\noindent
\textbf{Shift perturbation effect:}
We also illustrate the limitation  of our algorithm and how to avoid it in Section \ref{sim4}. 
The target probability distribution is $P_d(\xi\textbf{1}-\dif x)$
where $\textbf{1}=(1,\ldots, 1)\in\mathbb{R}^d$ 
and
\begin{equation}\nonumber
\xi=0,1,2,3,\ \mathrm{or}\ 4
\end{equation}
and $P_d$ is 
\begin{enumerate}
\item[(a)] the standard normal distribution, or 
\item[(b)] the $t$-distribution (two degrees of freedom). 
\end{enumerate}
We plot
\begin{enumerate}
\item[(ii)] the autocorrelation plot of
 $\left\{(2d)^{-1/2}(\|X_m^d-\xi\textbf{1}\|^2-d)\right\}_m$ for the standard normal distribution, 
and plot that of $\left\{\|X_m^d-\xi\textbf{1}\|^2/d\right\}_m$ for the $t$-distribution
for $\xi\in \{0,1,2,3,4\}$. 
\end{enumerate}
Although we can not apply our theoretical results in this non-spherically 
symmetric target distribution, it is a good example to illustrate the limitation of our algorithm. 
The performance of MCMC for the shift $\xi\textbf{1}$ will illustrate shift sensitivity of the MCMC algorithms. 
The RWM algorithms are, essentially, free from the shift. However the pCN and MpCN are sensitive for this effect. 
Fortunately, this effect can be avoided by simple estimate of the peak. 
We will show the results with and without this peak estimation. 

Since RWM algorithm is free from this effect, we only consider the pCN and MpCN algorithms. 
We can compare the results in this section to that of the RWM algorithms in Sections \ref{sim1} and \ref{sim2}. 
We set $\rho=0.8$ and set $X_0^d\sim N_d(0,I_d)$. 

\subsubsection{The Standard normal distribution in $\mathbb{R}^{20}$}\label{sim1}

Set $P_d=N_d(0,I_d)$ for $d=20$. 
For this case, the optimal convergence rate for the RWM algorithm
is $d$, and the Gaussian proposal distribution attains this rate (Theorem 3.1 of \citet{arXiv:1406.5392}). 
On the other hand, the pCN and MpCN  algorithms attains consistency and so these algorithms are better than the 
optimal RWM algorithm (Theorems  \ref{Theorem-1} and  \ref{Theorem-2}). 
The simulation shows that the performance of the RWM algorithm for the Gaussian proposal and the $t$-distribution proposal are similar (Figures \ref{Fig2} and \ref{Fig3}), and that for the pCN  and MpCN algorithms are also similar (Figures \ref{Fig4} and \ref{Fig5})
and are much better than the former two algorithms.  

\begin{figure}[H]
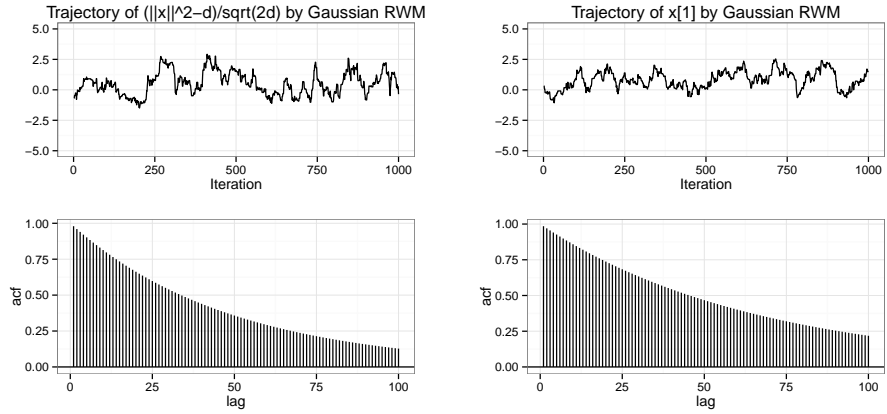

\centering
\plot{0.35}{360}{360}{num}
\plot{0.35}{360}{360}{num} 
\caption{The RWM algorithm with Gaussian proposal distribution for $P_d=N_d(0,I_d)$ for $d=20$.\label{Fig2}}
\end{figure}

\begin{figure}[H]
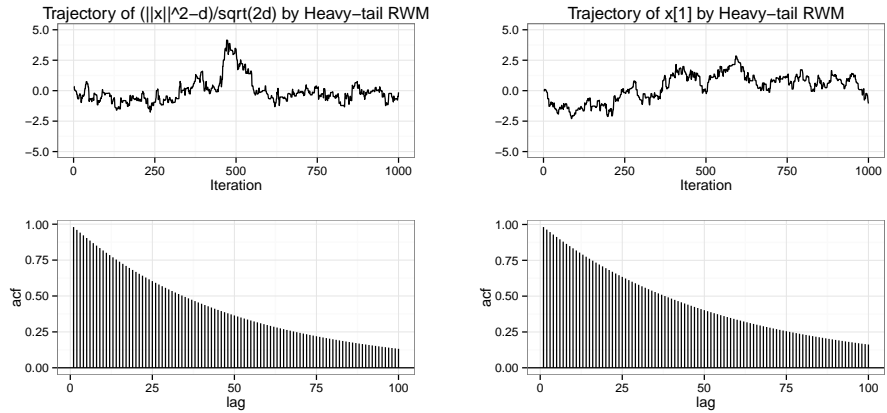

\centering
\plot{0.35}{360}{360}{num}
\plot{0.35}{360}{360}{num} 
\caption{The RWM algorithm with $t$-distribution as the proposal distribution for $P_d=N_d(0,I_d)$ for $d=20$.\label{Fig3}}
\end{figure}

\begin{figure}[H]
\centering
\plot{0.35}{360}{360}{num}
\plot{0.35}{360}{360}{num} 
\caption{The pCN algorithm for $P_d=N_d(0,I_d)$ for $d=20$.\label{Fig4}}
\end{figure}

\begin{figure}[H]
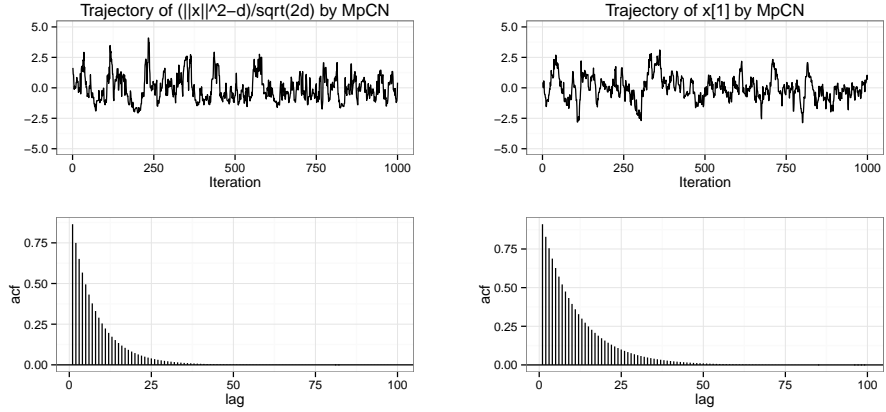

\centering
\plot{0.35}{360}{360}{num}
\plot{0.35}{360}{360}{num} 
\caption{The MpCN algorithm for $P_d=N_d(0,I_d)$ for $d=20$.\label{Fig5}}
\end{figure}

\subsubsection{$P_d$ is the $t$-distribution with two degrees of freedom in $\mathbb{R}^{20}$}\label{sim2}
Set $P_d$ as the $t$-distribution with $\nu=2$ degrees of freedom 
with the location parameter $\mu=0$ and the scale parameter $\sigma=5$ for $d=20$. 
Recall that the probability distribution  function is given by
\begin{equation}\nonumber
p_d(x)=\frac{\Gamma((\nu+d)/2)}{\Gamma(\nu/2)\nu^{d/2}\pi^{d/2}\sigma^d(1+\|(x-\mu)/\sigma\|^2/\nu)^{(\nu+d)/2}}. 
\end{equation}
For this case, the optimal convergence rate for the RWM algorithm
is $d^2$, and the Gaussian proposal distribution attains this rate (Theorem 3.2 of \citet{arXiv:1406.5392}). 
The pCN algorithm is much worse than the rate, and the  MpCN algorithm attains much better rate $d$ (Theorems \ref{Theorem-1} and \ref{Theorem-3}). 
In simulation, the MpCN algorithm (Figure \ref{Fig9}) is much better than other algorithms (Figures \ref{Fig6}-\ref{Fig8}) which corresponds to the theoretical result. 

\begin{figure}[H]
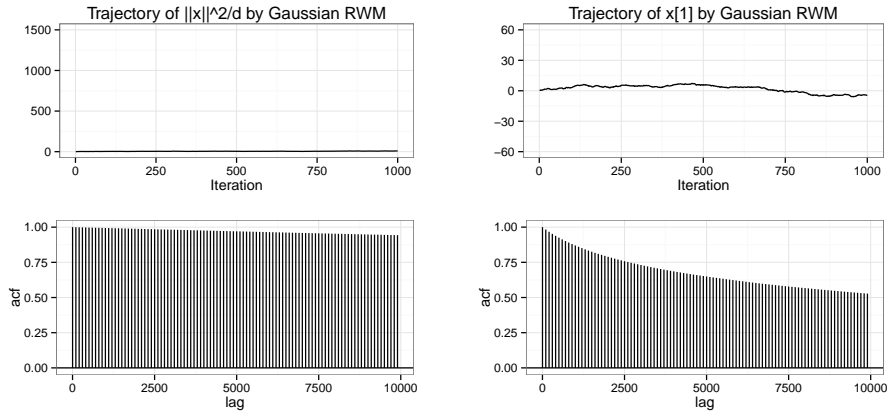

\centering
\plot{0.35}{360}{360}{num}
\plot{0.35}{360}{360}{num} 
\caption{The RWM algorithm with Gaussian proposal distribution when $t$-distribution  is the target distribution. \label{Fig6}}
\end{figure}

\begin{figure}[H]
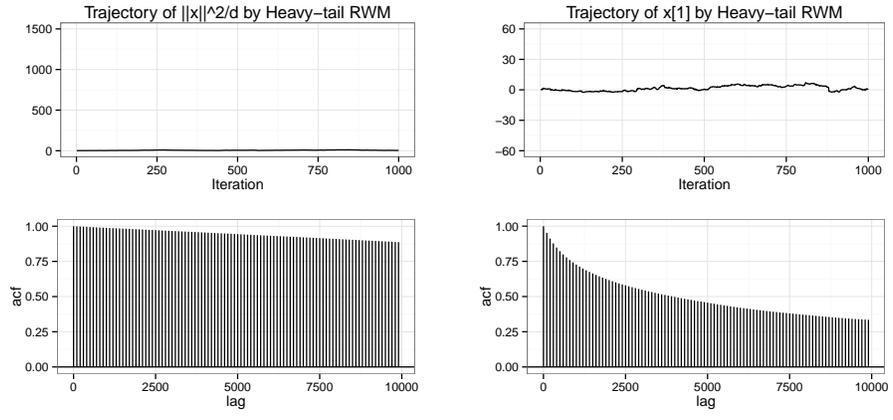

\centering
\plot{0.35}{360}{360}{num}
\plot{0.35}{360}{360}{num} 
\caption{The RWM algorithm with $t$-distribution as the proposal distribution and the target distribution is also the $t$-distribution.\label{Fig7}}
\end{figure}

\begin{figure}[H]
\centering
\plot{0.35}{360}{360}{num}
\plot{0.35}{360}{360}{num} 
\caption{The pCN algorithm when $t$-distribution  is the target probability distribution. \label{Fig8}}
\end{figure}

\begin{figure}[H]
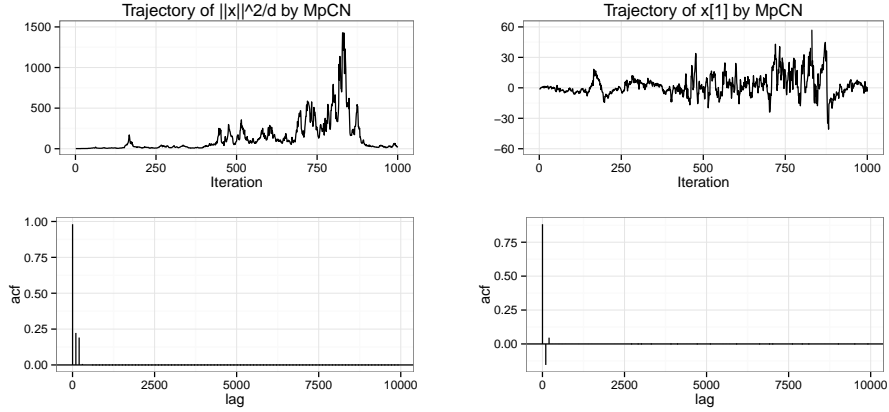

\centering
\plot{0.35}{360}{360}{num}
\plot{0.35}{360}{360}{num} 
\caption{The MpCN algorithm when $t$-distribution  is the target probability distribution.\label{Fig9}}
\end{figure}

\subsubsection{A perturbation of the $t$-distribution}\label{sim3}
We show the performance of the MpCN algorithm when the target distribution is not spherically symmetric. 
Let $P_{20}$ be a probability measure in $\mathbb{R}^{20}$ with the probability distribution function
\begin{equation}\nonumber
p_{20}(x_1,x_2,\ldots,x_{20})\propto\left(1+\sum_{i=1}^{20}\left(\frac{x_i-1}{5}\right)^2+|x_1|+\sin(x_2)/2\right)^{-(4+20)/2}. 
\end{equation}
The distribution is not scaled mixture and so we can not say 
anything for the convergence rate for this case. However by simulation we  observe that the MpCN algorithm (Figure \ref{Fig13})
is much better than other algorithms (Figures \ref{Fig10}-\ref{Fig12}). 

\begin{figure}[H]
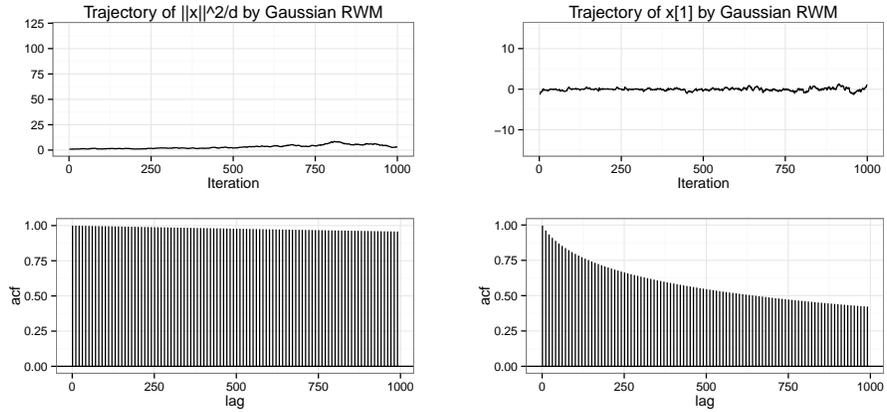

\centering
\plot{0.35}{360}{360}{num}
\plot{0.35}{360}{360}{num} 
\caption{The RWM algorithm with Gaussian proposal distribution when the perturbed $t$-distribution  is the target probability distribution. \label{Fig10}}
\end{figure}

\begin{figure}[H]
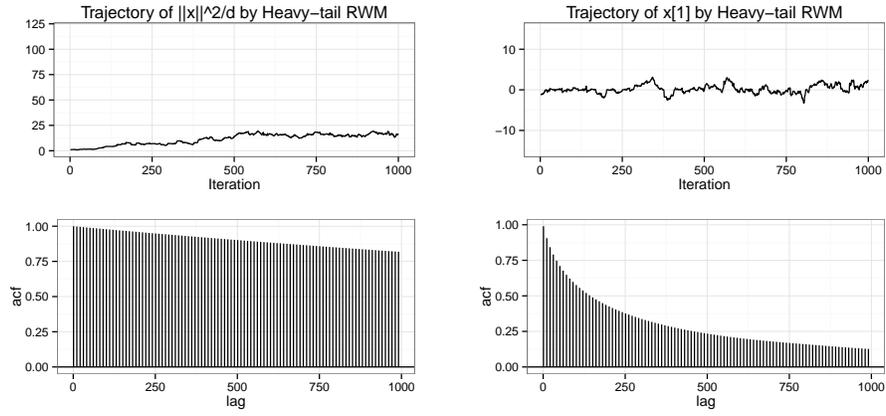

\centering
\plot{0.35}{360}{360}{num}
\plot{0.35}{360}{360}{num} 
\caption{The RWM algorithm with $t$-distribution as the proposal distribution and the target probability distribution is the perturbed $t$-distribution.\label{Fig11}}
\end{figure}

\begin{figure}[H]
\centering
\plot{0.35}{360}{360}{num}
\plot{0.35}{360}{360}{num} 
\caption{The pCN algorithm when the perturbed $t$-distribution  is the target probability distribution. \label{Fig12}}
\end{figure}

\begin{figure}[H]
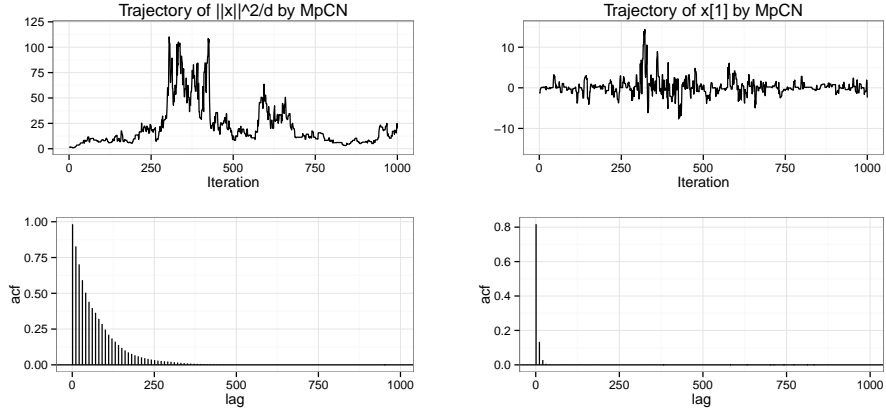

\centering
\plot{0.35}{360}{360}{num}
\plot{0.35}{360}{360}{num} 
\caption{The MpCN algorithm when the perturbed $t$-distribution  is the target probability distribution. \label{Fig13}}
\end{figure}

\subsubsection{Shift-perturbation of spherically symmetric target distributions}\label{sim4}

Let $P_d=N_d(\xi \textbf{1},I_d)$, where $\xi=0,1,2,3,4$ for $d=20$
and consider the pCN and MpCN algorithms. 
Compare the results of the RWM algorithms in Section \ref{sim1} (bottom left figures of Figures \ref{Fig2} and \ref{Fig3}). 
Figure \ref{Fig14} illustrates that although the performances of pCN and MpCN algorithms are much better than the RWM algorithms when $\xi=0$, 
it is sensitive to the value of $\xi$. Therefore for the light-tail target distribution in high-dimension, when the high-probability region is far from the origin, it is important to shift the target distribution in advance. 
For example, first, calculate rough estimate $\hat{\xi}$ of the peak of the target distribution $P_d(\dif x)$, and then 
run the MCMC algorithm for $P_d(-\hat{\xi}+\dif x)$. 
Some tempering strategy might be useful for the rough estimate of the peak
as  in \citet{UK}.  

\begin{figure}[H]
\centering
 \includegraphics[width=0.48\textwidth,bb=0 0 576 288]{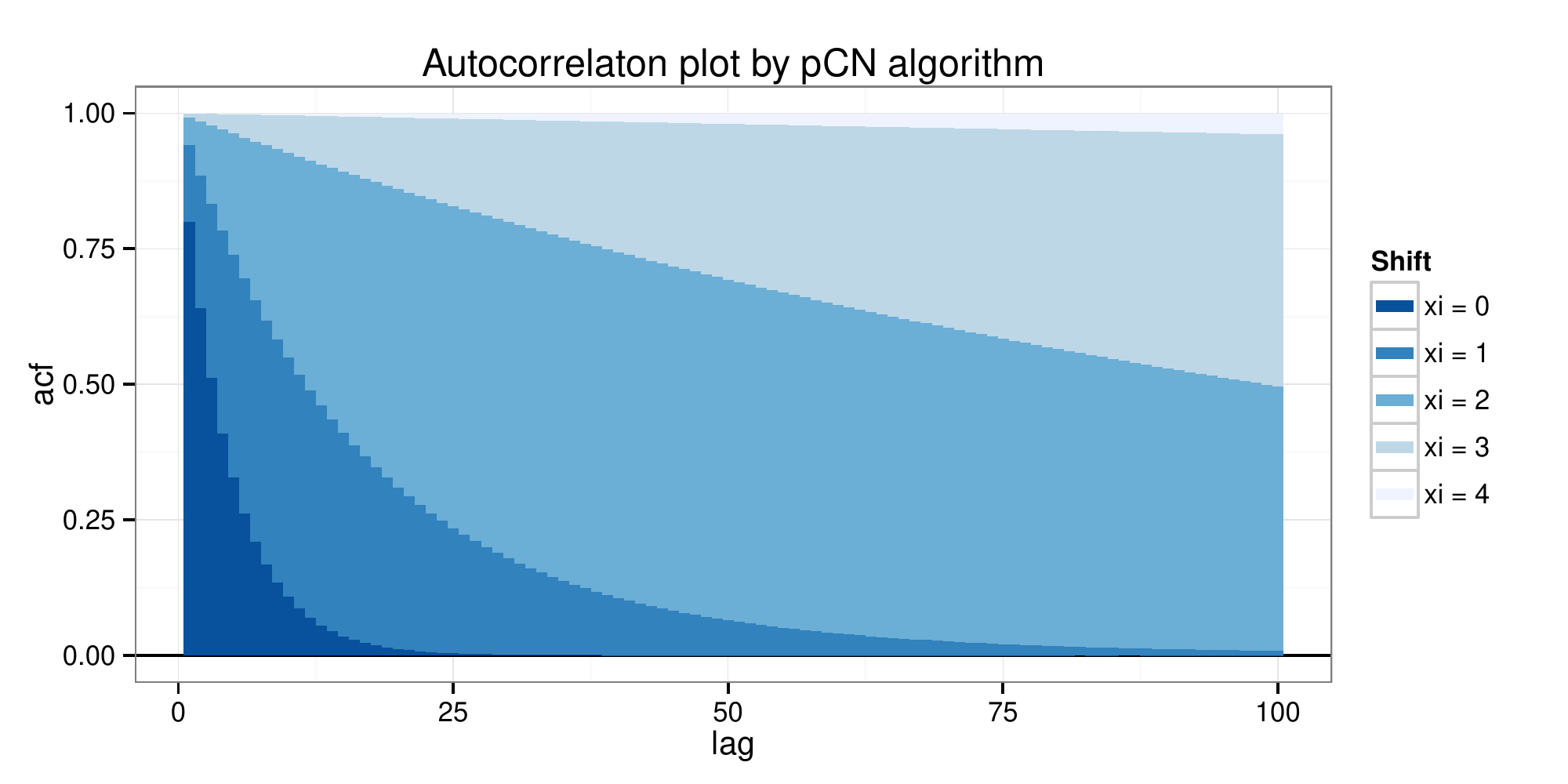}\includegraphics[width=0.48\textwidth,bb=0 0 576 288]{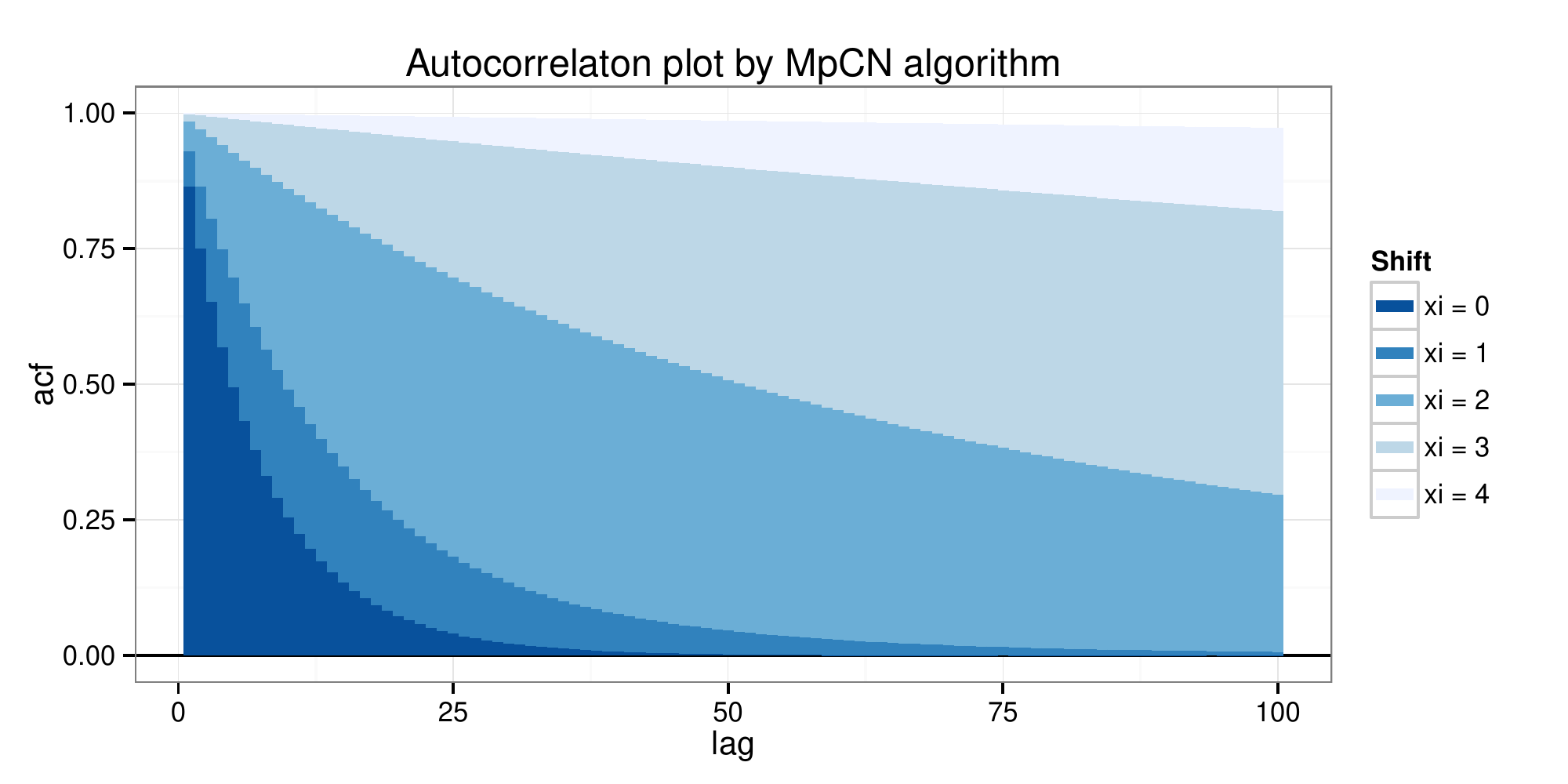}
  \caption{Autocorrelation plots for the pCN and MpCN algorithms for shifted normal distributions.\label{Fig14}}
\end{figure}

Next figure (Figure \ref{Fig15}) is a result of the pCN and MpCN algorithm with a simple peak estimation. We run $M=10^3$ iteration 
of the pCN or MpCN algorithm to calculate
\begin{equation}\label{j}
\hat{\xi}=M^{-1}\sum_{m=0}^{M-1}X_m^d
\end{equation}
and then run $M=10^8$ iteration 
of the pCN or MpCN algorithm for the target probability distribution $P_d(-\hat{\xi}+\dif x)$. The effect of the shift is considerably weakened. 

\begin{figure}[H]
\centering
 \includegraphics[width=0.48\textwidth,bb=0 0 576 288]{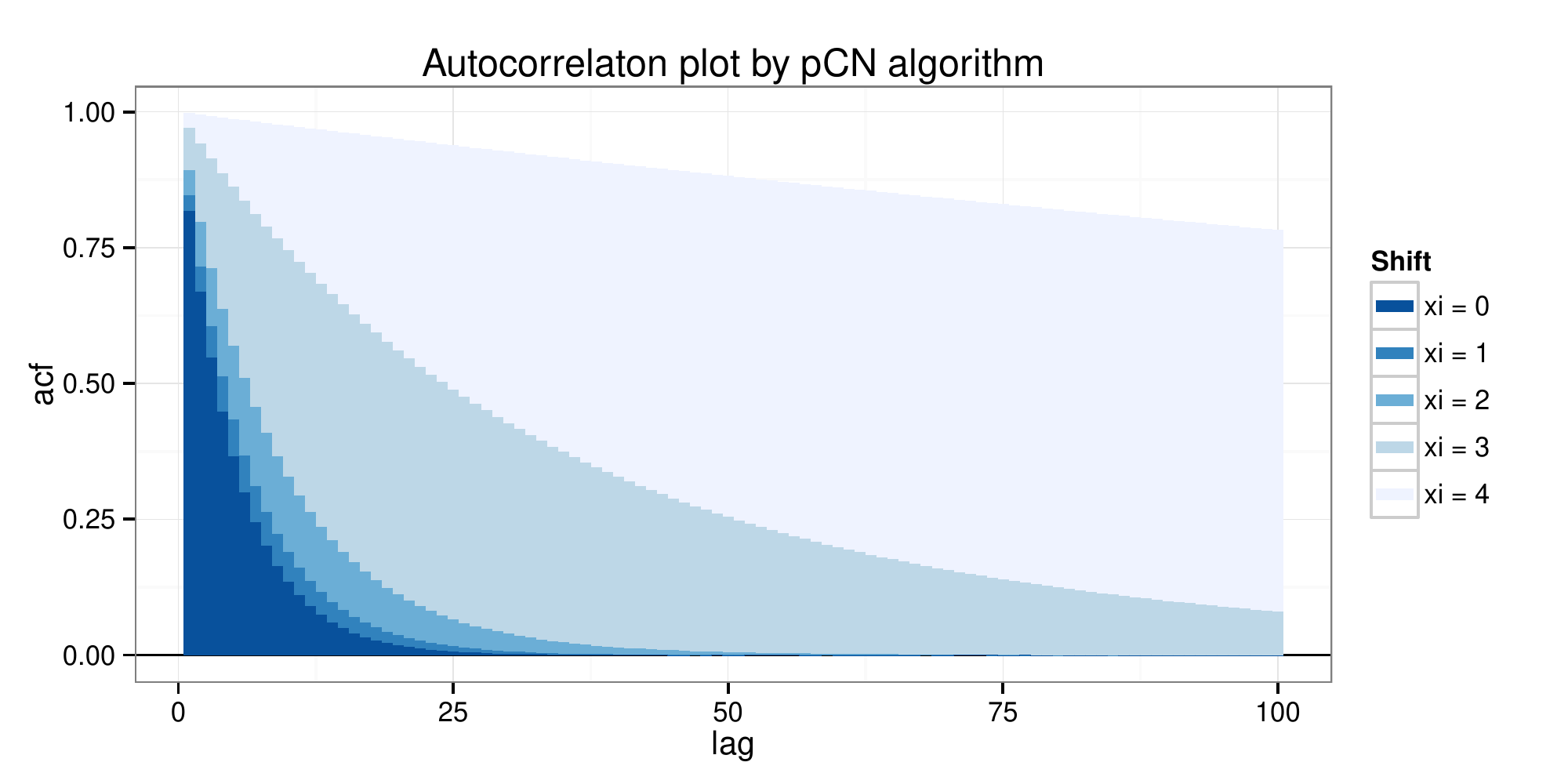}\includegraphics[width=0.48\textwidth,bb=0 0 576 288]{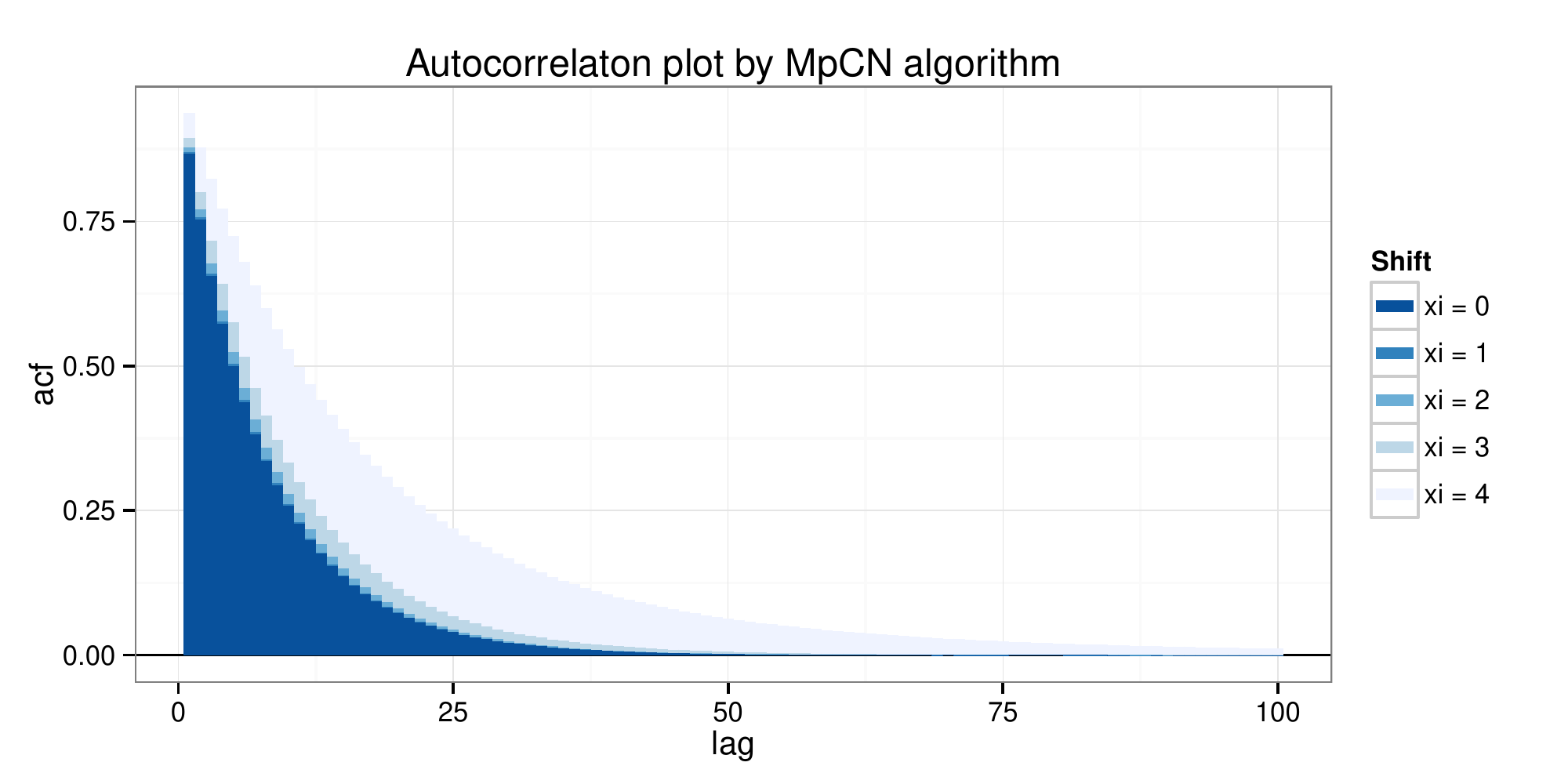}
  \caption{Autocorrelation plots for the pCN and MpCN algorithms for shifted normal distributions
  with an initial estimate of the peak.\label{Fig15}}
\end{figure}

We consider the $t$-distribution 
with $\nu=2$ and 
$\sigma=5$ where  $\xi=0,1,2,3,4$ for $d=20$. 
Compare the results in  Section \ref{sim2} for the RWM algorithms (bottom left figures of Figures \ref{Fig6} and \ref{Fig7}). 
Compared to the light-tailed distribution, the effect of the shift is small for the MpCN algorithm
and the five autocorrelation plots are overlapped in Figure \ref{Fig16}. 

\begin{figure}[H]
\centering
 \includegraphics[width=0.48\textwidth,bb=0 0 576 288]{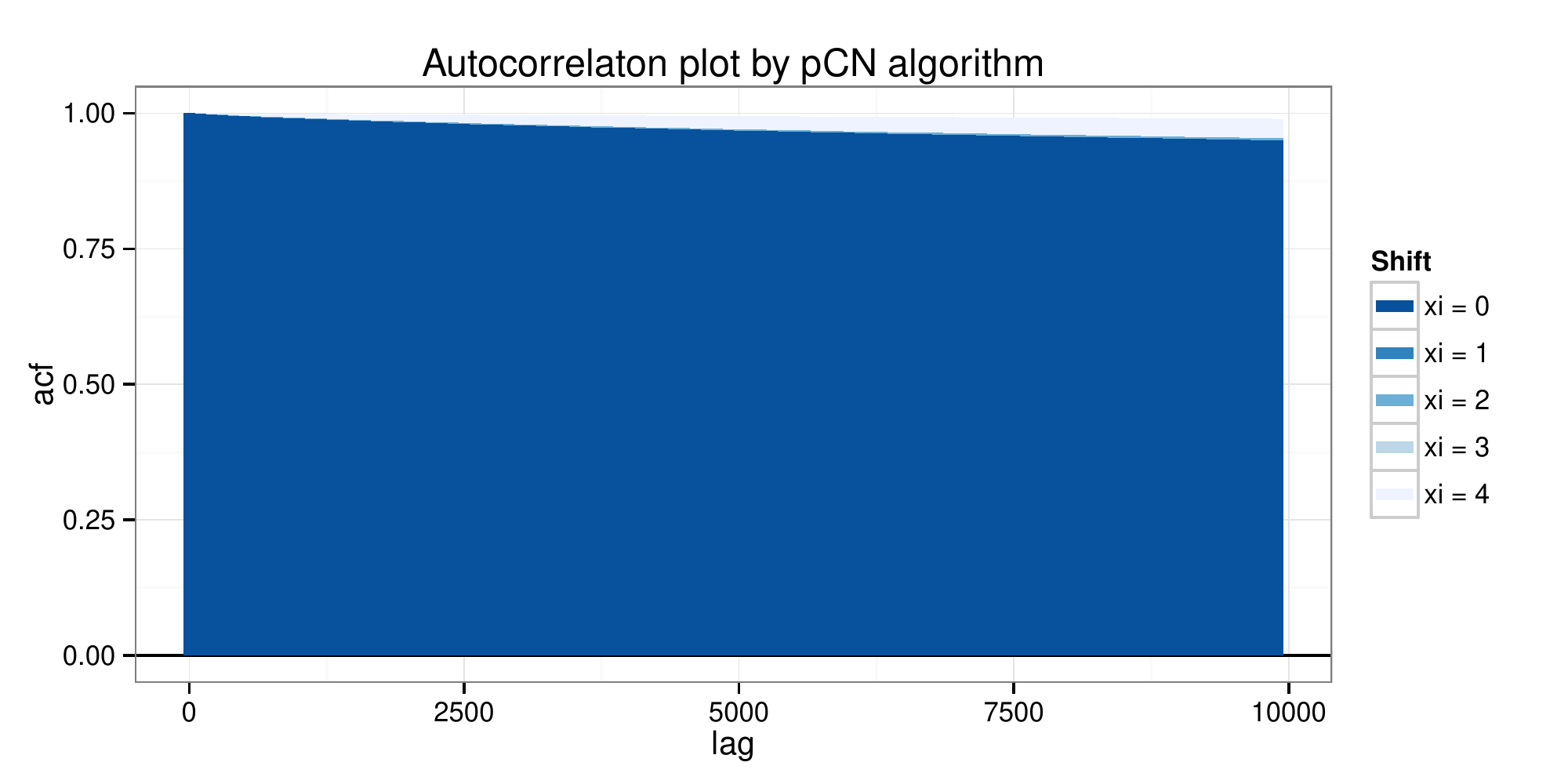}
  \includegraphics[width=0.48\textwidth,bb=0 0 576 288]{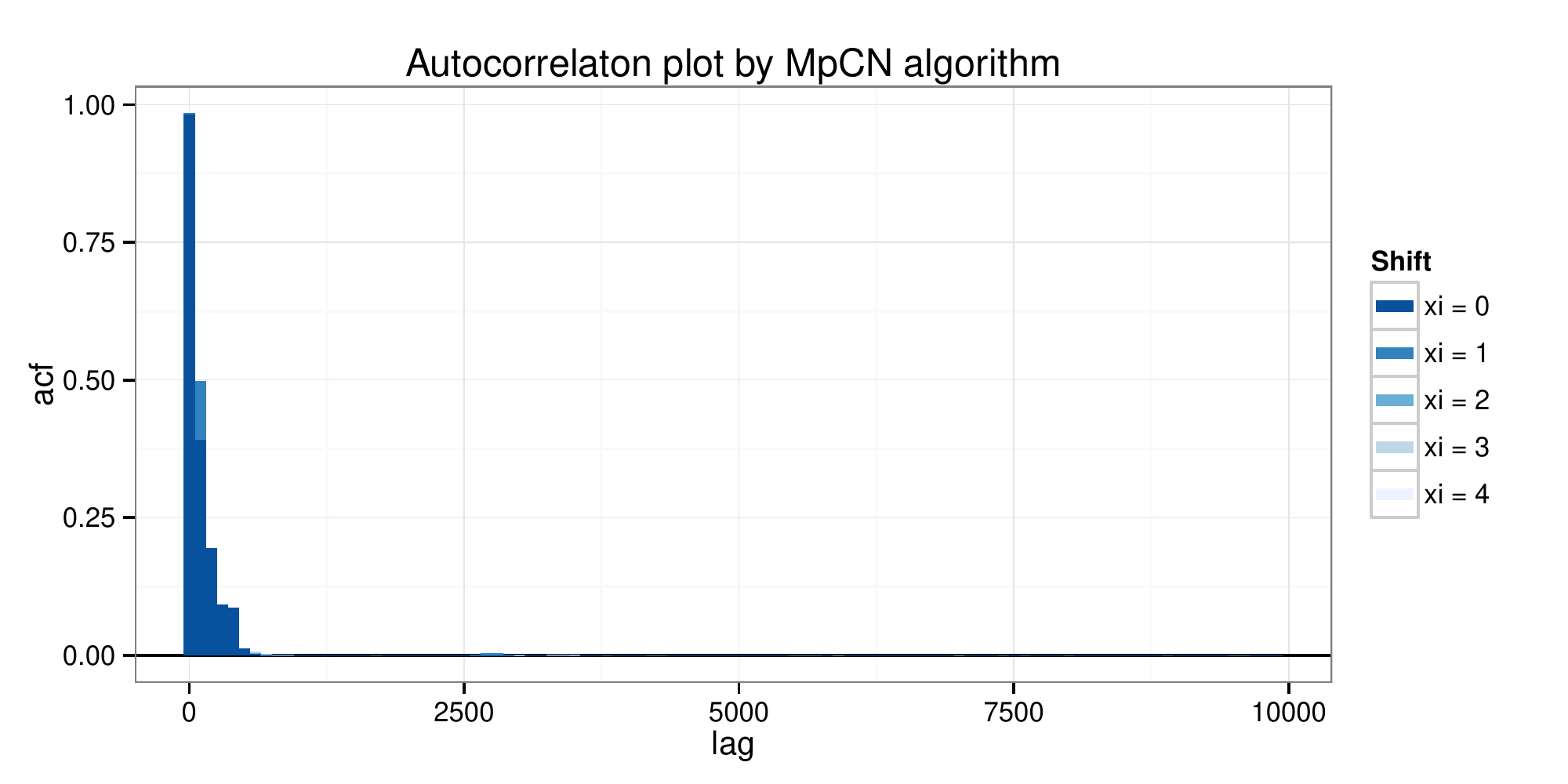}
  \caption{Autocorrelation plots for the pCN and MpCN algorithms for shifted $t$-distributions.\label{Fig16}}
\end{figure}

The next figure (Figure \ref{Fig17}), which is almost identical to the previous one,  is a result  of $M=10^8$ iteration of the pCN and MpCN algorithm with a simple peak estimation (\ref{j}) by $M=10^3$ iteration. 
Thus for heavy-tailed target distribution, the effect of shift perturbation is small, and the gain of 
the peak estimation is also small. 

\begin{figure}[H]
\centering
 \includegraphics[width=0.48\textwidth,bb=0 0 576 288]{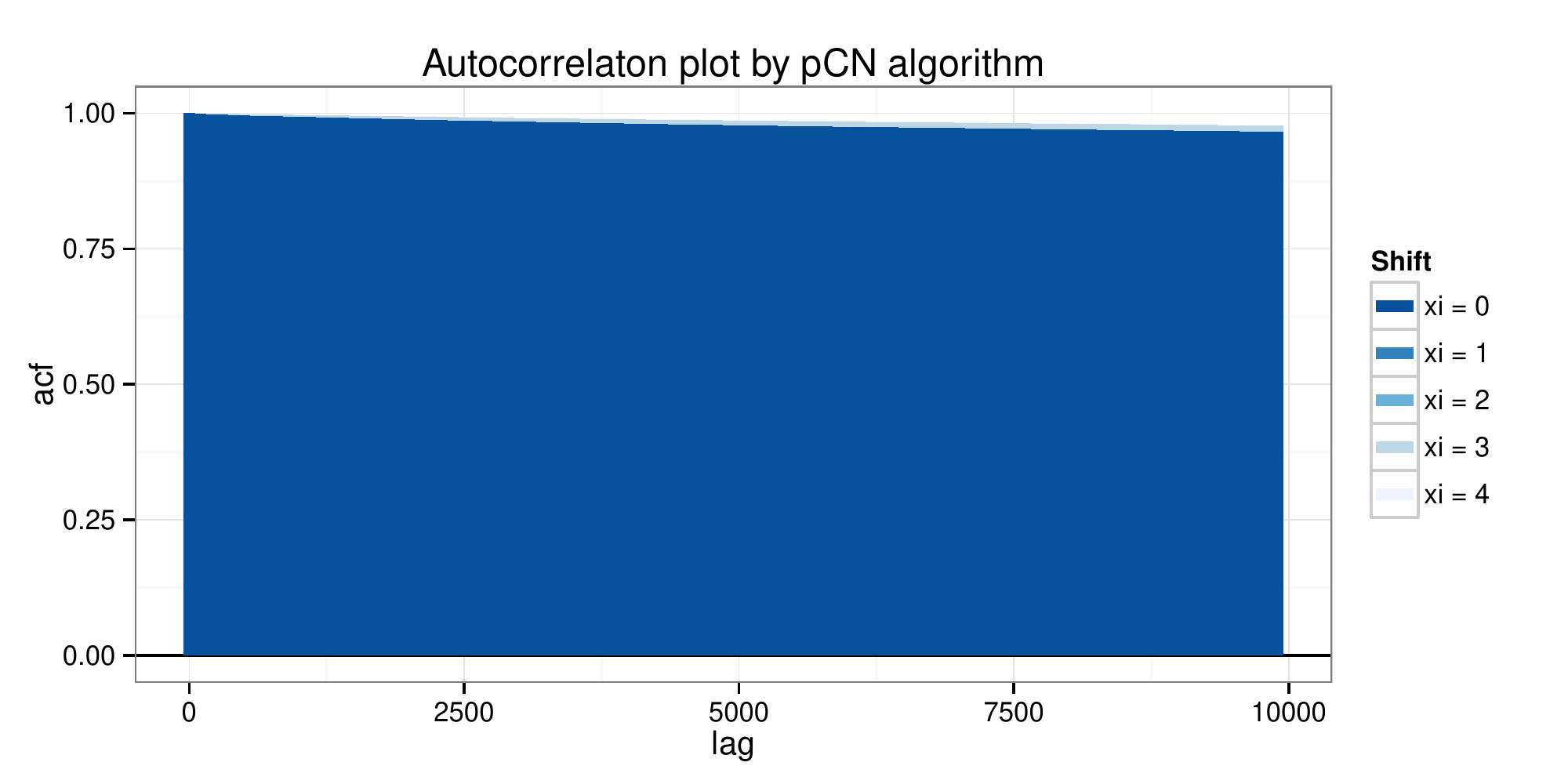}
  \includegraphics[width=0.48\textwidth,bb=0 0 576 288]{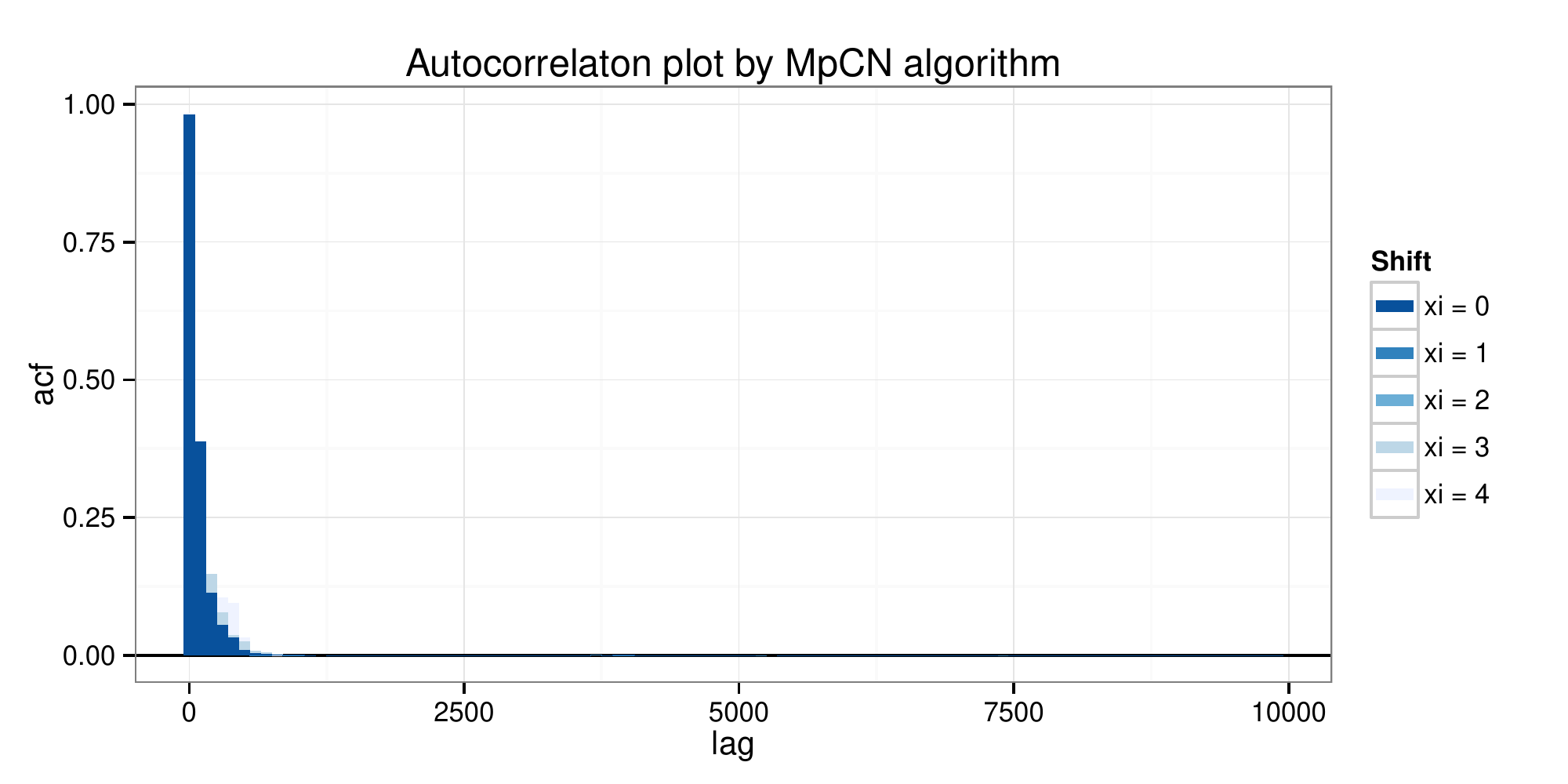}
  \caption{Autocorrelation plots for the pCN and MpCN algorithms for shifted $t$-distributions
  with an initial estimate of the peak. \label{Fig17}}
\end{figure}

\section{High-dimensional asymptotic theory}\label{highdimsection}

We consider a sequence of the target probability distributions $\{P_d\}_{d\in\mathbb{N}}$ indexed by the number of dimension $d$. 
For a given $d$, $P_d$ is a $d$-dimensional probability distribution that is a scale mixture of the normal distribution.  
Furthermore, our asymptotic setting is that the number of dimension $d$ goes infinity while the mixing distribution $Q$ of $P_d$ is unchanged. 

Note that in our results, stationarity and reversibility are essential. However this can be weakened 
as explained in Lemma 4 of \cite{Kamatani10}. 

\subsection{Consistency}

In this section, we review consistency of MCMC studied   in \cite{Kamatani10}. 
Set a sequence of Markov chains $\xi^d:=\left\{\xi^d_m;m\in\mathbb{N}_0\right\}\ (d\in\mathbb{N})$ with the invariant probability distributions  $\left\{\Pi_d\right\}_d$. 
The law of $\xi^d$ is called consistent if 
\begin{equation}\label{consistency}
\frac{1}{M}\sum_{m=0}^{M-1}f(\xi^d_m)-\Pi_d(f)=o_{\mathbb{P}}(1)
\end{equation}
for any $M, d\rightarrow\infty$ for any bounded continuous function $f$. This  says that 
the integral $\Pi_d(f)$ we want to calculate is approximated by a Monte Carlo simulated value
$\frac{1}{M}\sum_{m=0}^{M-1}f(\xi^d_m)$ after a reasonable number of iteration $M$.  
For example, regular Gibbs sampler should satisfy this type of property (more precisely, local consistency. See \cite{Kamatani10})
when $d$ is the sample size of the data. 
In the current case, the state space for  $X^d=\left\{X^d_m;m\in\mathbb{N}_0\right\}\ (d\in\mathbb{N})$
 changes as $d\rightarrow\infty$ that is inconvenient for further analysis.  
As in \citet{arXiv:1406.5392}, to overcome the difficulty, we set a projection $\pi_E=\pi_{d,E}$ for a finite subset $E\subset\left\{1,\ldots, d\right\}$ by
 \begin{equation}\nonumber
 \pi_E(x)=(x_i)_{i\in E}\ (x=(x_i)_{i=1,\ldots, d}). 
 \end{equation}
 
 \begin{definition}[Consistency]
 We call that the law of a $\mathbb{R}^d$-valued Markov chain $\left\{X^d_m\right\}_{m\in\mathbb{N}_0}$ is consistent if 
 \begin{equation}\label{eq1}
\frac{1}{M_d}\sum_{m=0}^{M_d-1}f\circ\pi_{E_d^k}(X^d_m)-P_d(f\circ\pi_{E_d^k})=o_{\mathbb{P}}(1)
 \end{equation}
as $d\rightarrow\infty$ for any 
$k\in\mathbb{N}$, $M_d\rightarrow\infty$ and for any bounded continuous function $f:\mathbb{R}^k\rightarrow\mathbb{R}$
and any $k$-elements $E_d^k$ of  $\left\{1,\ldots, d\right\}$. 
 \end{definition}
  We write $\pi_k$ for $\pi_{\left\{1,\ldots, k\right\}}$. In \citet{arXiv:1406.5392}, the role of $E_d^k$ is important, but in this paper, 
 we can assume that $\pi_{E_d^k}=\pi_k$ throughout in this paper due to rotational symmetricity  of the pCN and MpCN algorithms. 
As in \citet{arXiv:1406.5392}, we relax the condition for $M_d$ and introduce the convergence rate. 
 
  \begin{definition}[Weak Consistency]
We call that the law of a $\mathbb{R}^d$-valued Markov chain $\left\{X^d_m\right\}_{m\in\mathbb{N}_0}$ is weakly consistent with rate $T_d$ if 
(\ref{eq1}) is satisfied
 for any $M_d\rightarrow\infty$ such that 
 $M_d/T_d\rightarrow\infty$. 
 We will call the rate $T_d$, the convergence rate. 
 If $T_d/d^k\rightarrow 0$ for some $k\in\mathbb{N}$, we call that it has a polynomial rate of convergence. 
\end{definition}

The rate $T_d$ corresponds to the number of iteration until good convergence. Therefore smaller is better. 
In \citet{arXiv:1406.5392}, we showed that the optimal rate for the RWM algorithm is $d^2$
for heavy-tailed target probability distribution. We will show that this rate becomes $d$ for the MpCN algorithm. 
Note that when the MCMC is consistent, the convergence rate is $T_d=1$.

\subsection{Assumption for the target probability distribution}
  
Let $Q(\dif z)$ be a probability measure on $(0,\infty)$. 
Let $P_d$ be a scale mixture of the normal distribution defined by 
\begin{equation}\nonumber
P_d=\mathcal{L}(X^d_0),\ Q_d=\mathcal{L}(\|X^d_0\|^2/d)
\end{equation}
where $X^d_0|Z\sim N_d(0,ZI_d)$ and $Z\sim Q$. 
We will write
\begin{equation}\label{d}
P_d=P_d(Q). 
\end{equation}
In particular, $P_d(\delta_{\left\{1\right\}})=N_d(0,I_d)$.
Note that $Q_d\Rightarrow Q$ as $d\rightarrow\infty$ since $\|X_0^d\|^2/d\rightarrow Z\ \mathrm{a.s.}$
In this setup, $P_d$ and $Q_d$ have probability distribution functions $p_d$ and $q_d$ that satisfy 
\begin{equation}\nonumber
p_d(x)\propto \|x\|^{2-d} q_d\left(\frac{\|x\|^2}{d}\right). 
\end{equation} 
If $P_d=P_d(Q)$, the acceptance ratio of the MpCN algorithm defined in 
(\ref{MpCN}) can be written in the following form: 
\begin{equation}\label{accept_ratio}
\alpha_d(x,y)=\min\left\{1,\frac{\tilde{q}_d\left(\frac{\|y\|^2}{d}\right)}{\tilde{q}_d\left(\frac{\|x\|^2}{d}\right)}\right\}
\end{equation} 
where $\tilde{q}_d(x)=xq_d(x)$. 
We will assume the following regularity condition on $Q$ to show some properties of the MpCN algorithm. 

\begin{assumption}\label{Assumption-1}
Probability distribution $Q$ has the strictly positive continuously differentiable probability distribution function $q(y)$. 
Each $q(y)$ and $q'(y)$ vanishes at $+0$ and $+\infty$. 
\end{assumption}

\begin{lemma}[Lemma 2.2 of \citet{arXiv:1406.5392}]\label{assumplem}
Under Assumption \ref{Assumption-1},  $\lim_{d\rightarrow\infty}\|q_d-q\|_\infty=0$
and $\lim_{d\rightarrow\infty}\|q_d'-q'\|_\infty=0$. 
\end{lemma}

Probability distribution $P_d(Q)$ that satisfies above includes many heavy-tailed probability distributions such as the $t$-distribution and the stable distribution. 
See \citet{arXiv:1406.5392}. 

%

\subsection{Main results}

Even for Gaussian target probability distribution,  the pCN algorithm may not work well. 
If the target probability distribution $P_d(Q)$ is different from $N_d(0,I_d)$, then 
any polynomial number of iteration is not sufficient for the pCN algorithm to have a
good approximation of the integral we want to calculate. 

\begin{theorem}\label{Theorem-1}
Let $P_d=P_d(Q)$. 
Then $\mathrm{pCN}(P_d)$ have the polynomial rate of convergence if and only if $Q=\delta_{\{1\}}$. If $Q=\delta_{\{1\}}$, then $\mathrm{pCN}(P_d)$ has consistency. 
\end{theorem}

\begin{proof}
The results for $Q\neq\delta_{\{1\}}$ comes from Lemma \ref{inconsistency}. 
Consistency for $Q=\delta_{\{1\}}$ comes from Lemma \ref{consistencylem}. 
\end{proof}

On the other hand, the MpCN algorithm always works well for light-tailed target distribution. 
More precisely, the following holds: 

\begin{theorem}\label{Theorem-2}
$\mathrm{MpCN}(P_d)$ is consistent for $P_d=N_d(0,\sigma^2 I_d)$ for any $\sigma>0$. 
\end{theorem}

\begin{proof}
By considering consistency of $\left\{X_m^d/\sigma\right\}_m$, 
it is sufficient to prove for $\sigma=1$, 
which is proved in Lemma \ref{consistencylem}. 
\end{proof}

When $P_d$ is a heavy-tailed distribution, we already know that the pCN algorithm does not work well by Theorem \ref{Theorem-1}. However, the MpCN algorithm still has a good convergence property. 
Recall that the optimal convergence rate for the RWM algorithm is $d^2$ as studied in \citet{arXiv:1406.5392}. 
Let $[x]$ denote the integer part of $x>0$. 
See Section \ref{discuss} for the proof of Proposition \ref{Proposition-1} and Theorem \ref{Theorem-3}.

\begin{proposition}\label{Proposition-1}
Let $Q$ satisfy 
Assumption \ref{Assumption-1} and $P_d=P_d(Q)$. 
Set  $X^d\sim \mathrm{MpCN}(P_d)$, and let $Y_t^d=r_d(X^d_{[dt]})$. Then $Y^d=(Y^d_t)_t$ converges to 
the stationary ergodic process $Y=(Y_t)_t$ (in Skorohod's topology)
that is the solution of 
\begin{equation}\label{Proposition-1-eq1}
\dif Y_t=a(Y_t)\dif t+\sqrt{b(Y_t)} \dif W_t; Y_0\sim Q
\end{equation}
where
\begin{equation*}\nonumber
a(y)
=2(2y+(\log q)'(y)y^2)(1-\rho),\ 
b(y)=4y^2(1-\rho). 
\end{equation*}
\end{proposition}



\begin{theorem}\label{Theorem-3}
Let $Q$ satisfy 
Assumption \ref{Assumption-1} and $P_d=P_d(Q)$. 
Then
 $\mathrm{MpCN}(P_d)$ has the convergence rate $d$. 
\end{theorem}

\subsection{Discussion}\label{discuss}

\begin{itemize}
\item In \citet{arXiv:1406.5392}, we defined optimality among all the RWM algorithms. 
In the current study, it is difficult to find suitable sense of optimality. Na\"ive sense of optimality may not work. 
We can find a rather impractical MCMC that is consistent  for any $P_d(Q)$ by making a mixture of the MpCN algorithm and independent type Metropolis-Hastings algorithm with the  proposal probability distribution $P_d(Q^*)$ for any $Q^*$ that satisfies Assumption \ref{Assumption-1}. 
To construct a practically useful sense of optimality is an open problem. 
I believe that the MpCN algorithm  has a kind of optimality. 
\item Proposal transition kernel used in MpCN has the invariant  distribution $\overline{P}_d$ defined in (\ref{stationprop}), and so this is a special case of MCMC that uses reversible proposal transition kernel. 
The relation to the target probability distribution $P_d$ and $\overline{P}_d$ is quite important. 
If $\overline{P}_d$ has a heavier-tail than that of $P_d$, then MCMC behaves relatively well. 
On the other hand, if $\overline{P}_d$ has a lighter-tail, it becomes quite poor. 
The RWM algorithm has $\overline{P}_d=\mathrm{Uniform\ distribution}$. This is a robust choice, but it loses  efficiency to pay the price as described in  \citet{arXiv:1406.5392}. 
On the other hand, the pCN algorithm, which has $\overline{P}_d=N_d(0,I_d)$,  does not work well except 
some specific cases. The proposed algorithm, MpCN is in the middle of these algorithms. It is robust and works well. 
\item It is possible to consider a more general class of the MpCN algorithm:
Let $\overline{Q}$ be a $\sigma$-finite measure on $(0,\infty)$ and set
 $\overline{P}_d=P_d(\overline{Q})$ with density $\overline{p}_d$. 
For $m\ge 1$, set
\begin{equation}\nonumber
\left\{\begin{array}{l}
Z^d_m\sim \phi_d(X^d_{m-1};0,zI_d)\overline{Q}(\dif z)\\
X^{d*}_m= \sqrt{\rho}X^d_{m-1}+\sqrt{(1-\rho)Z_m^d}W^d_m,\ W^d_m\sim N_d(0,I_d)\\
X^d_m=\left\{\begin{array}{ll}
X^{d*}_m & \mathrm{with\ probability}\ \alpha_d(X^d_{m-1},X^{d*}_m)\\
       X^d_{m-1} & \mathrm{with\ probability}\ 1-\alpha_d(X^d_{m-1},X^{d*}_m)
       \end{array}\right.
\end{array}\right. 
\end{equation}
where $\alpha_d(x,y)=\min\left\{1, p_d(y)\overline{p}_d(x)/p_d(x)\overline{p}_d(y)\right\}$
assuming that $\int \phi_d(x;0,zI_d)\overline{Q}(\dif z)<\infty$ for any $x\in\mathbb{R}^d$. 
For example, in \cite{UK}, $\overline{Q}(\dif z)\propto z^{-\nu/2-1}e^{-\nu/(2z)}$. 
If $\overline{Q}$ satisfies Assumption \ref{Assumption-1}, then this algorithm has the same asymptotic property as our  MpCN algorithm
and our algorithm is a special case $\overline{Q}(\dif z)=z^{-1}\dif z$. 
We believe that the choice of $\overline{Q}$ has a little effect in practice. 
\item There is no theoretical results for the MpCN algorithm for target probability distributions with shift perturbation discussed in Section \ref{sim4}. It might be possible to study scaling limit theorem for this direction. 
\item The class of target probability distributions we considered is quite restrictive. 
The extension of the class is not straightforward and probably it requires some 
new techniques. 
However as illustrated in simulation, we believe that by using our restrictive class, we successfully described the real behaviour of the MCMC algorithms  and it will be surprising if we find a completely different story 
by generalising this class. 
\end{itemize}

\section{Proofs}\label{proofsec}

%
Let $K_\delta=[\delta,\delta^{-1}]$ for $\delta\in (0,1)$. 
\subsection{Consistency results for Gaussian target probability distribution}
By definition, the pCN algorithm defined in (\ref{pCN}) has the following form:
\begin{align}
\|X_1^{d*}\|^2-\|X_0^d\|^2&=
-(1-\rho)\|X_0^d\|^2+2\sqrt{\rho(1-\rho)}\left\langle W_1^d, X_0^d\right\rangle 
+(1-\rho) \|W_1^d\|^2. \label{crgeq0}
\end{align}
The MpCN algorithm defined in (\ref{MpCN}) has a similar form
\begin{align}
\|X_1^{d*}\|^2-\|X_0^d\|^2&=
-(1-\rho)\|X_0^d\|^2+2\sqrt{\rho(1-\rho)Z_1^d}\left\langle W_1^d, X_0^d\right\rangle 
+(1-\rho)Z_1^d \|W_1^d\|^2\label{crgeq1}\\
&=2\sqrt{\rho(1-\rho)Z_1^d}\left\langle W_1^d, X_0^d\right\rangle 
+(1-\rho)Z_1^d \left(\|W_1^d\|^2-\frac{\|X_0^d\|^2}{Z_1^d}\right).\nonumber
\end{align}
The conditional law of $\|X_0^d\|^2/Z_1^d$ given $X_0^d$
is chi-squared distribution with $d$ degrees of freedom and so we will write
$\|\tilde{W}_1^d\|^2=\|X_0^d\|^2/Z_1^d$
where $\tilde{W}_1^d\sim N_d(0,I_d)$. By this notation, the above becomes 
\begin{equation}\label{crgeq2}
\|X_1^{d*}\|^2-\|X_0^d\|^2
=2\sqrt{\rho(1-\rho)Z_1^d}\left\langle W_1^d, X_0^d\right\rangle 
+(1-\rho)Z_1^d \left(\|W_1^d\|^2-\|\tilde{W}_1^d\|^2\right).
\end{equation}

\begin{lemma}\label{consistencylem}
If $P_d=N_d(0,I_d)$, then $\mathrm{pCN}(P_d)$ and $\mathrm{MpCN}(P_d)$  are consistent. 
\end{lemma}

\begin{proof}
When $P_d=N_d(0,I_d)$, the $k$-dimensional process $\left\{\pi_k(X_m^d)\right\}_m$ is a Markov chain for 
the pCN algorithm. Moreover, it is a stationary ergodic $\mathrm{AR}(1)$ process
since the acceptance ratio is $\alpha_d(x,y)\equiv 1\ (x,y\in\mathbb{R}^d)$ for this case. 
Since the law of this $\mathrm{AR}(1)$ process does not depend on $d$, 
consistency of $\mathrm{pCN}(P_d)$ comes from classical point-wise ergodic theorem. 

Now we assume $X^d\sim \mathrm{MpCN}(P_d)$ and set 
\begin{equation}\nonumber
r_d(x)=(2d)^{-1/2}(\|x\|^2-d)\ (x\in\mathbb{R}^d). 
\end{equation}
For this case, $\left\{\pi_k(X_m^d)\right\}_m$ is not a Markov chain, but  $\left\{\left(r_d(X_m^d),\pi_k(X_m^d)\right)\right\}_m$ is a Markov chain. 
Set
\begin{equation}\nonumber
\left\{\begin{array}{l}
R_m^d=r_d(X_m^d)\\
R_m^{d*}=r_d(X_m^{d*})
\end{array}, 
\right. 
\left\{\begin{array}{l}
S_m^d=\pi_k(X_m^d)\\
S_m^{d*}=\pi_k(X_m^{d*})
\end{array}. 
\right. 
\end{equation}
To prove consistency by using Lemma \ref{lem1}, we need to show weak convergence of the process $\left\{(R_m^d,S_m^d)\right\}_m$ to the limit process $\left\{(R_m,S_m)\right\}_m$  defined below, and 
ergodicity of this limit: for $m\ge 1$, 
\begin{equation}\nonumber
\left\{\begin{array}{l}
\left(\begin{array}{c}
R_m^*\\
S_m^*
\end{array}\right)= 
\left(\begin{array}{c}
R_{m-1}+\sqrt{2(1-\rho)}W_m^R,\\
\sqrt{\rho}S_{m-1}+\sqrt{1-\rho}W_m^S
\end{array}\right), \ W_m^R\sim N(0,1),\ W_m^S\sim N_k(0,I_k),\\
\left(\begin{array}{c}
R_m\\
S_m
\end{array}\right)=\left\{\begin{array}{ll}
\left(\begin{array}{c}
R_m^*\\
S_m^*
\end{array}\right)
 & \mathrm{with\ probability}\ \alpha(R_{m-1},R_m^*),\\
\left(\begin{array}{c}
R_{m-1}\\
S_{m-1}
\end{array}\right) & \mathrm{with\ probability}\ 1-\alpha(R_{m-1},R_m^*),
       \end{array}\right.
\end{array}\right. 
\end{equation}
where $\alpha(x,y)=\min\left\{1,\exp(-y^2/2+x^2/2)\right\}$
and $R_0\sim N(0,1)$, $S_0\sim N_k(0,I_k)$. 
Note that the Markov chain $\{(R_m,S_m)\}_m$ has the same law as that generated by the Metropolis-Hastings algorithm 
with the target probability distribution $N_{k+1}(0,I_{k+1})$. 
 First we prove the weak convergence. 
It
 can be proved by total variation convergence 
\begin{equation}\label{e}
\|\mathcal{L}(R_0^d, S_0^d, R_1^{d*}, S_1^{d*})-
\mathcal{L}(R_0, S_0, R_1^*, S_1^*)\|_{\mathrm{TV}}\rightarrow 0
\end{equation}
by Lemmas \ref{MHlemma1} and \ref{MHlemma2}
since both $\{(R_m^d, S_m^d)\}_m$ and $\{(R_m,S_m)\}_m$ are generated by Metropolis-Hastings algorithms. 
For (\ref{e}), it is sufficient to show 
\begin{equation}\label{crgeq3}
\lim_{d\rightarrow\infty}\|\mathcal{L}(F_d)-N_{2k+2}\|_{\mathrm{TV}}= 0\ \mathrm{where}\ 
F_d=\left(R_0^d, S_0^d, \frac{R_1^{d*}-R_0^d}{\sqrt{2(1-\rho)}},\frac{S_1^{d*}-\sqrt{\rho}S_0^d}{\sqrt{1-\rho}}\right)
\end{equation}
where, by (\ref{crgeq2}),  
\begin{align*}
\frac{R_1^{d*}-R_0^d}{\sqrt{2(1-\rho)}}&=\sqrt{\rho Z_1^d}d^{-1/2}\left\langle W_1^d, X_0^d\right\rangle 
+\sqrt{\frac{1-\rho}{2}}Z_1^d \left(r_d(W_1^d)-r_d(\tilde{W}_1^d)\right)\\
&=\sqrt{\rho Z_1^d}d^{-1/2}\left\langle W_1^d, X_0^d\right\rangle 
+\sqrt{\frac{1-\rho}{2}}\left(Z_1^d r_d(W_1^d)-(Z_1^d-1)r_d(\tilde{W}_1^d)-r_d(\tilde{W}_1^d)\right),\\
\frac{S_1^{d*}-\sqrt{\rho}S_0^d}{\sqrt{1-\rho}}&=\sqrt{Z_1^d}\pi_k(W^d_1). 
\end{align*}
We decompose $F_d$ into the sum of the following random variables:
\begin{align*}\nonumber
F_{d,1}&=\left(R_0^d,S_0^d, -\sqrt{\frac{1-\rho}{2}} r_d(\tilde{W}_1^d),0\right),\\
F_{d,2}&=\left(0, 0, \sqrt{\rho Z_1^d}d^{-1/2}\left\langle W_1^d, X_0^d\right\rangle 
+\sqrt{\frac{1-\rho}{2}}\left(Z_1^d r_d(W_1^d)-(Z_1^d-1)r_d(\tilde{W}_1^d)\right),\sqrt{Z_1^d}\pi_k(W^d_1)\right). 
\end{align*}
We show
\begin{equation}\label{mixconv}
\mathcal{L}(F_{d,1})\Rightarrow N_{2k+2}(0, A),\ \mathcal{L}(F_{d,2}|F_{d,1})\Rightarrow N_{2k+2}(0, B)\ \mathrm{in\ probability},\ 
\end{equation}
where 
\begin{equation}\nonumber
A= \begin{pmatrix}
  1 & 0 & 0 & 0 \\
  0 & I_k & 0 & 0 \\
  0  & 0  & \frac{1-\rho}{2} & 0  \\
  0 & 0 & 0 & 0
 \end{pmatrix},\ 
B=
 \begin{pmatrix}
  0 & 0 & 0 & 0 \\
  0 & 0 & 0 & 0 \\
  0  & 0  & \frac{1+\rho}{2} & 0  \\
  0 & 0 & 0 & I_k
 \end{pmatrix}. 
\end{equation}
The former convergence of (\ref{mixconv}) is an easy conclusion of Slutsky's lemma. 
For the latter, we use  Skorohod's representation theorem.
By using this,  we can assume $F_{d,1}(\omega)\rightarrow F_1(\omega)\sim N_{2k+2}(0, A)$
for each $\omega$, and $\sigma(F_{d,1})$-measurable random variables 
$Z_1^d-1$ and $(Z_1^d-1)r_d(\tilde{W}_1^d)$ converges to $0$ for each $\omega$. 
Then the latter convergence of (\ref{mixconv})
also comes from  Slutsky's lemma. 

Observe that the random variable $F_{d,1}$, and the random variable $F_{d,2}$ conditioned on $F_{d,1}$ are composed by the first and the second Wiener chaoses. Therefore, by Theorem \ref{NourdinPoly},  convergences in (\ref{mixconv}) also imply 
the total variation convergences. 
Hence the law of $F_d=F_{d,1}+F_{d,2}$ converges in total variation to $N_{2k+2}$, 
and therefore, the weak convergence of 
$\{(R_m^d,S_m^d)\}_m$ follows. 

Finally we prove the ergodicity of the process $\left\{(R_m,S_m)\right\}_m$.  However this follows by 
Corollary 2 of  \citet{TierneyAOS94} and hence the claim follows.

\end{proof}

\subsection{Inconsistency for the pCN algorithm}

%

In this and subsequent section, we set
\begin{equation}\nonumber
r_d(x)=\frac{\|x\|^2}{d}\ (x\in\mathbb{R}^d)
\end{equation}
and write $R_m^d=r_d(X_m^d)$ and $R_m^{d*}=r_d(X_m^{d*})$. 
Let $\mathrm{int}(A)$ be the interior of a set $A$. 

\begin{lemma}\label{inconsistlem0}
Let $P_d=P_d(Q)$. For $\left\{X_m^d\right\}_m\sim\mathrm{pCN}(P_d)$, for any $p\in\mathbb{N}$ and any compact subset $K$ of $(0,\infty)\backslash\left\{1\right\}$, we have
\begin{align*}
d^p\mathbb{P}(R_0^d\in K, X_0^d\neq X_1^d)=o(1). 
\end{align*}
\end{lemma}

\begin{proof}
Let $\overline{R}_m^d=r_d(X_m^d)-1$ and $\overline{R}_m^{d*}=r_d(X_m^{d*})-1$
be the ``centered'' version of $R_m^d$ and $R_m^{d*}$. 
Let $I_0$ and $I_1$ be any compact subsets of $(-1,\infty)\backslash\left\{0\right\}$ such that 
\begin{equation}\label{include}
\rho I_0\subset \mathrm{int}(I_1)\ \mathrm{or}\ I_1\subset \mathrm{int}(\rho I_0)
\end{equation}
where $\rho A=\left\{\rho x; x\in A\right\}$. 
For the former case in (\ref{include}),  for $\epsilon=\mathrm{dist}(\rho I_0,I_1^c)=\inf\left\{|x-y|; x\in \rho I_0, y\notin I_1\right\}$
we will show
\begin{equation}\label{crgeq4}
d^p\mathbb{P}\left(\overline{R}_0^d\in I_0,\ \overline{R}_1^d\in I_1^c,\ X_0^d\neq X_1^d\right)=o(1).
\end{equation}
For the latter case, 
 for $\epsilon=\mathrm{dist}((\rho I_0)^c,I_1)=\inf\left\{|x-y|; x\notin \rho I_0, y\in I_1\right\}$
we will prove
\begin{equation}\label{crgeq5}
d^p\mathbb{P}\left(\overline{R}_0^d\in I_0^c,\ \overline{R}_1^d\in I_1,\ X_0^d\neq X_1^d\right)=o(1).
\end{equation}
On the event $\left\{X_0^d\neq X_1^d\right\}$, we have $X_1^{d*}=X_1^d$. 
By (\ref{crgeq0}) we have
\begin{equation}\nonumber
\overline{R}_1^{d*}=\rho\overline{R}_0^d+2\frac{\|X_0^d\|}{d}\sqrt{\rho(1-\rho)}\left\langle \frac{X_0^d}{\|X_0^d\|},W_1^d\right\rangle+(1-\rho)\left(\frac{\|W_1^d\|^2}{d}-1\right).
\end{equation}
Therefore, on the events in the left-hand side of (\ref{crgeq4}) or (\ref{crgeq5}), we have
\begin{equation}\label{crgeq6}
d^{1/2}\epsilon\le
d^{1/2}\left|\overline{R}_1^{d*}-\rho\overline{R}_0^d\right|
\le
2\left(\frac{\|X_0^d\|^2}{d}\right)^{1/2}\left|\left\langle\frac{X_0^d}{\|X_0^d\|},W_1^d\right\rangle\right|
+d^{1/2}\left|\frac{\|W_1^d\|^2}{d}-1\right|. 
\end{equation}
On the event in (\ref{crgeq4}), we have
$\|X_0^d\|^2/d=\overline{R}_0^d+1\le \sup_{x\in I_0}|x|+1$, 
and on the event in (\ref{crgeq5}), by triangular inequality together with  (\ref{pCN}), we have 
\begin{align*}
\sqrt{\rho}\left(\frac{\|X_0^d\|^2}{d}\right)^{1/2}\le 
\left(\frac{\|X_1^{d*}\|^2}{d}\right)^{1/2}+
\sqrt{1-\rho}\left(\frac{\|W_1^d\|^2}{d}\right)^{1/2}
\le (\sup_{x\in I_1}|x|+1)^{1/2}+\left(\frac{\|W_1^d\|^2}{d}\right)^{1/2}. 
\end{align*}
Observe that $\left\langle\frac{X_0^d}{\|X_0^d\|},W_1^d\right\rangle\sim N(0,1)$. 
Together with this fact and  Proposition \ref{contractive}, the right-hand side of (\ref{crgeq6}) is bounded above by a random variable (say) $\eta_d$ such that  $\sup_d\mathbb{E}\left[\eta_d^q\right]^{1/q}<\infty$
for any $q\in\mathbb{N}$. 
Therefore (\ref{crgeq4}) follows since 
\begin{align*}
\mathbb{P}\left(\overline{R}_0^d\in I_0,\ \overline{R}_1^d\in I_1^c,\ X_0^d\neq X_1^d\right)
&=
\mathbb{P}\left(\overline{R}_0^d\in I_0,\ \overline{R}_1^{d*}\in I_1^c,\ X_0^d\neq X_1^d\right)\\
&\le 
\mathbb{P}\left(d^{1/2}\epsilon\le \eta_d\right)\le  \mathbb{E}\left[\left(\frac{\eta_d}{d^{1/2}\epsilon}\right)^{q}\right]=O(d^{-q/2})=o(d^{-p})
\end{align*}
by Chevyshev's inequality by taking $q>2p$. In the same way, 
(\ref{crgeq5}) can be proved. 
Now, choose compact  subsets $I_0, I_1, I_2$ so that
\begin{equation}\nonumber
I_1\subset \mathrm{int}(\rho I_0),\ \rho I_1\subset \mathrm{int}(I_2),\ I_0\cap I_2=\emptyset. 
\end{equation}
For example, for $\epsilon\in (0,1)$, set $I_1=[a,b], I_0=[\rho^{-1+\epsilon}a, \rho^{-1-\epsilon}b]$ and $I_2=[\rho^{1+\epsilon} a,\rho^{1-\epsilon} b]$
so that $\rho^{2(1-\epsilon)}\le a/b$ if $a,b>0$ and $\rho^{2(1+\epsilon)}\le b/a$ if $a,b<0$.  
By (\ref{crgeq4}) and (\ref{crgeq5})  with reversibility by Lemma \ref{reverse}, 
\begin{align*}\nonumber
\mathbb{P}\left(\overline{R}_0^d\in I_1,\ \overline{R}_1^d\in (I_0\cap I_2)^c,\ X_0^d\neq X_1^d\right)
&\le
\mathbb{P}\left(\overline{R}_0^d\in I_1,\ \overline{R}_1^d\in I_0^c,\ X_0^d\neq X_1^d\right)
+\mathbb{P}\left(\overline{R}_0^d\in I_1,\ \overline{R}_1^d\in I_2^c,\ X_0^d\neq X_1^d\right)\\
&=
\mathbb{P}\left(\overline{R}_0^d\in I_0^c,\ \overline{R}_1^d\in I_1,\ X_0^d\neq X_1^d\right)
+\mathbb{P}\left(\overline{R}_0^d\in I_1,\ \overline{R}_1^d\in I_2^c,\ X_0^d\neq X_1^d\right)\\
&=o(d^{-p}). 
\end{align*}
However, since $I_0\cap I_2=\emptyset$, the above proves
\begin{align*}\nonumber
\mathbb{P}\left(\overline{R}_0^d\in I_1,\ X_0^d\neq X_1^d\right)
=o(d^{-p}). 
\end{align*}
Since any compact set can be covered by  finite family of the compact sets $I_1$, the claim follows. 
\end{proof}

\begin{lemma}\label{inconsistency}
For $P_d=P_d(Q)$, $\mathrm{pCN}(P_d)$  does not have any polynomial rate of convergence if $Q(\left\{1\right\})<1$. 
\end{lemma}

\begin{proof}
By assumption, there exists a compact set $K$ such that $K\subset(0,\infty)\backslash\left\{1\right\}$ and $Q(\mathrm{int}(K))\ge \delta$ for $\delta>0$. By $Q_d\Rightarrow Q$ and 
 by Lemma \ref{inconsistlem0}, for any $p\in\mathbb{N}$, 
\begin{align*}\nonumber
\liminf_{d\rightarrow\infty}\mathbb{P}\left(\forall i,j<d^p, X_i^d= X_j^d\right)
&\ge  
\liminf_{d\rightarrow\infty}\mathbb{P}\left(R_0^d\in \mathrm{int}(K), \forall i,j<d^p, X_i^d= X_j^d\right)\\
&\ge 
\liminf_{d\rightarrow\infty}\mathbb{P}(R_0^d\in \mathrm{int}(K))- \limsup_{d\rightarrow\infty}\mathbb{P}\left(R_0^d\in K, \exists i,j<d^p, X_i^d\neq X_j^d\right)\\
&\ge 
Q(\mathrm{int}(K))- \limsup_{d\rightarrow\infty}d^p\mathbb{P}(R_0^d\in K, X_0^d\neq X_1^d)=Q(\mathrm{int}(K))\ge \delta. 
\end{align*}
Thus we have the following degenerate property:
\begin{align*}\nonumber
\liminf_{d\rightarrow\infty}\mathbb{P}\left(\frac{1}{d^p}\sum_{m=0}^{d^p-1}f\circ\pi_1(X_m^d)=f\circ\pi_1(X_0^d)\right)\ge \delta
\end{align*}
for any bounded continuous function $f(x)$ where $\pi_1(x)=x_1$ is the first component of the vector $x=(x_1,\ldots, x_d)\in\mathbb{R}^d$. 

Assume by the way of contradiction that $X^d\sim \mathrm{pCN}(P_d)$ is weakly consistent with rate $T_d$ where $T_d/d^p\rightarrow 0$. Then the following should also be satisfied: 
\begin{align*}\nonumber
\frac{1}{d^p}\sum_{m=0}^{d^p-1}f\circ\pi_1(X_m^d)-P_d(f\circ\pi_1)=o_{\mathbb{P}}(1). 
\end{align*}
Recall that $P_d=P_d(Q)$ is the scale mixture of the normal distribution as defined in (\ref{d}). 
By these two convergence properties together with the fact $\mathcal{L}(\pi_1(X_0^d))= P_1(Q)$, 
we have
\begin{align*}\nonumber
P_1(Q)(\left\{x;|f(x)-P_1(Q)(f)|<\epsilon\right\})=\liminf_{d\rightarrow\infty}\mathbb{P}(|f\circ\pi_1(X_0^d)-P_d(f\circ\pi_1)|<\epsilon)\ge \delta
\end{align*}
for any $\epsilon>0$. 
By monotone convergence theorem, this is possible only if  $P_1(Q)(\left\{x; f(x)=c\right\})\ge \delta$ for some $c\in\mathbb{R}$, and thus 
it is not satisfied for example, for $f(x)=\arctan(x)$
since $P_1(Q)$ has a probability density function. Therefore $\mathrm{pCN}(P_d)$ cannot be weakly consistent with rate $T_d$ where $T_d/d^p\rightarrow 0$ for any $p>0$ and hence $\mathrm{pCN}(P_d)$ cannot have polynomial rate of convergence. 
\end{proof}

\subsection{Convergence of the MpCN algorithm  for heavy-tail case}\label{cma}

Let $\left\{X_m^d\right\}_m\sim\mathrm{MpCN}(P_d)$. As in the previous section, we set 
$
r_d(x)=\frac{\|x\|^2}{d}\ (x\in\mathbb{R}^d)
$
and write 
\begin{equation}\label{cmaeq1}
\left\{\begin{array}{l}
R_m^d=r_d(X_m^d)\\
R_m^{d*}=r_d(X_m^{d*})
\end{array}, 
\right. 
\left\{\begin{array}{l}
S_m^d=\pi_k(X_m^d)\\
S_m^{d*}=\pi_k(X_m^{d*})
\end{array}. 
\right. 
\end{equation}
By definition, 
\begin{equation}\label{r1star}
R_1^{d*}=R_0^d\left(1+2d^{-1/2}\sqrt{1-\rho}F_d\right)
\end{equation}
where	
\begin{equation}\label{Fd}
F_d:=d^{1/2}\frac{\|X_1^{d*}\|^2-\|X_0^d\|^2}{2\|X_0^d\|^2\sqrt{1-\rho}}=
d^{1/2}\frac{R_1^{d*}-R_0^d}{2R_0^d\sqrt{1-\rho}}. 
\end{equation}
Write $y=R_0^d$. 
Write $\mathbb{P}_y$ and $\mathbb{E}_y$ for the conditional probability and expectation given $y$. 

\begin{proof}[Proof of Proposition \ref{Proposition-1}]
We rewrite the acceptance ratio $\alpha_d(X_0^d, X_1^{d*})$ in (\ref{accept_ratio}) as $\alpha_{y,d}(F_d)$ by using $F_d$: 
\begin{equation}\nonumber
\left\{\begin{array}{l}
\beta_{y,d}(F_d):=\tilde{q}_d(R_1^{d*})/\tilde{q}_d(R_0^d)=\tilde{q}_d\left(y\left(1+2d^{-1/2}\sqrt{1-\rho}F_d\right)\right)/\tilde{q}_d(y),\\ 
\alpha_{y,d}(F_d):=\min\left\{1,\beta_{y,d}(F_d)\right\}. 
\end{array}\right.
\end{equation}
Let 
\begin{equation}\label{triplet}
\left\{\begin{array}{lll}
a_d(y)&=d\mathbb{E}_y[R_1^d-R_0^d],\\
b_d(y)&=d\mathbb{E}_y[(R_1^d-R_0^d)^2],\\
c_d(y)&=d\mathbb{E}_y[(R_1^d-R_0^d)^4]. 
\end{array}\right.
\end{equation}
We estimate the triplet. By representation (\ref{r1star}), we have 
\begin{equation}\label{f}
a_d(y)=
d\mathbb{E}_y\left[\left(R_1^{d*}-R_0^d\right)\alpha_{y,d}(F_d)\right]=
2yd^{1/2}\sqrt{1-\rho}\mathbb{E}_y[F_d\alpha_{y,d}(F_d)]. 
\end{equation}
We are going to estimate the expectation in the right-hand side by using Proposition \ref{main}. 
Note here that $\alpha_{y,d}'=\beta_{y,d}'1_{\left\{\beta_{y,d}<1\right\}}$ is not bounded since
\begin{align*}
\beta_{y,d}'(F_d)=2yd^{-1/2}\sqrt{1-\rho}\beta_{y,d}(F_d)(\log \tilde{q}_d)'(R_1^{d*}), 
\end{align*}
and $(\log \tilde{q}_d)'$ is not bounded in general. To overcome the difficulty, 
we put a tempering function $\phi_d:\mathbb{R}\rightarrow\mathbb{R}$ which is continuous and 
$\phi_d(x)=x$ if $|x|\le d^{1/4}$ and piecewise constant otherwise. 
By Lemma \ref{assumplem},  the tempered version $\tilde{\alpha}_{y,d}:=\alpha_{y,d}\circ\phi_d$
has a bounded derivative 
\begin{align*}\nonumber
\tilde{\alpha}_{y,d}'(F_d)&=\alpha_{y,d}'(F_d)1_{\{|F_d|<d^{1/4}\}}=\beta_{y,d}'(F_d)1_{\left\{\beta_{y,d}(F_d)<1, |F_d|<d^{1/4}\right\}}\\
&=2yd^{-1/2}\sqrt{1-\rho}\beta_{y,d}(F_d)(\log \tilde{q}_d)'(R_1^{d*})1_{\left\{\beta_{y,d}(F_d)<1, |F_d|<d^{1/4}\right\}}. 
\end{align*}
Moreover, $\sup_{y\in K_\delta}\|\tilde{\alpha}_{y,d}'\|_\infty=O(d^{-1/2})$
and $\sup_{y\in K_\delta}\|\tilde{\alpha}_{y,d}\|_\infty\le 1$
since $R_1^{d*}\in K_{\delta/2}$ for sufficiently large $d$ by (\ref{r1star}). 
Therefore we can apply Proposition \ref{main} to $f(x)=\tilde{\alpha}_{y,d}(x)$ and we have
\begin{align}\nonumber
\sup_{y\in K_\delta}\left|\mathbb{E}_y[F_d\tilde{\alpha}_{y,d}(F_d)]-\mathbb{E}_y[\tilde{\alpha}_{y,d}'(F_d)]-d^{-1/2}\sqrt{1-\rho}\mathbb{E}_y[\tilde{\alpha}_{y,d}(F_d)]\right|\\
=O\left(\max\left\{d^{-1/2}\sup_{y\in K_\delta}\|\tilde{\alpha}_{y,d}'\|_\infty,d^{-1}\sup_{y\in K_\delta}\|\tilde{\alpha}_{y,d}\|_\infty\right\}\right)=O(d^{-1}).\label{a}
\end{align}
Now we show uniform convergence (in $K_\delta$) of  the three expectations in the left-hand side in the above. 
The first expectation can be estimated by Chevyshev's inequality together with Lemma \ref{malliavinlemma}:
\begin{align*}
\left|\mathbb{E}_y\left[F_d\alpha_{y,d}(F_d)\right]-\mathbb{E}_y\left[F_d\tilde{\alpha}_{y,d}(F_d)\right]\right|\le \mathbb{E}_y\left[\left|F_d\right|,|F_d|\ge d^{1/4}\right]\le d^{-3/4}\mathbb{E}_y\left[|F_d|^4\right]
=O(d^{-3/4}).  
\end{align*}
For the uniform convergence of the second and third expectations in the left-hand side of (\ref{a}), 
suppose that $y_d\in K_\delta\ (d\in\mathbb{N})$, and so
without loss of generality, we can assume that there is a limit $y_d\rightarrow y^*\in K_\delta$. 
By Proposition \ref{steinprop}, 
\begin{align*}
|\mathbb{E}_{y_d}[\tilde{\alpha}_{y_d,d}'(F_d)]-N[\tilde{\alpha}_{y_d,d}']|&=
o(\sup_{y\in K_\delta}\|\tilde{\alpha}_{y,d}'\|_\infty)=o(d^{-1/2}),\\
|\mathbb{E}_{y_d}[\tilde{\alpha}_{y_d,d}(F_d)]-N[\tilde{\alpha}_{y_d,d}]|&=
o(\sup_{y\in K_\delta}\|\tilde{\alpha}_{y,d}\|_\infty)=o(1),
\end{align*}
as $d\rightarrow\infty$
where $Nf=\mathbb{E}[f(X)],\ X\sim N(0,1)$. 
By Lemma \ref{assumplem}, 
the following convergence (a.s. in the Lebesgue measure) is satisfied for each $f\in\mathbb{R}$
depending on whether  $(\log \tilde{q})'(y^*)>0$ or 
$(\log \tilde{q})'(y^*)<0$:
\begin{align*}
\lim_{d\rightarrow\infty}1_{\left\{\beta_{y_d,d}(f)<1\right\}}= 1_{\left\{f<0\right\}}, \mathrm{or}\  1_{\left\{f>0\right\}}
\end{align*}
where $\tilde{q}(y)=yq(y)$. 
Also $\lim_{d\rightarrow\infty}\beta_{y_d,d}(f)=1$ is satisfied. By Lebesgue's dominated convergence theorem, we have
\begin{align*}
\lim_{d\rightarrow\infty}d^{1/2}N[\tilde{\alpha}_{y_d,d}']= y^*\sqrt{1-\rho}(\log \tilde{q})'(y^*). 
\end{align*}
On the other hand, if $(\log \tilde{q})'(y^*)=0$, then
\begin{align*}
\lim_{d\rightarrow\infty}d^{1/2}N[\tilde{\alpha}_{y_d,d}']= 0=y^*\sqrt{1-\rho}(\log \tilde{q})'(y^*)
\end{align*}
by Lebesgue's dominated convergence theorem. 
In the same way, $\lim_{d\rightarrow\infty}N[\tilde{\alpha}_{y_d,d}(F_d)]= 1$, which completes to show uniform convergence of the three expectations in the left-hand side of (\ref{a}). 
These uniform convergences yield
\begin{align}\nonumber
\sup_{y\in K_\delta}\left|d^{1/2}\mathbb{E}_y[F_d\alpha_{y,d}(F_d)]-y\sqrt{1-\rho}(\log \tilde{q})'(y)-\sqrt{1-\rho}\right|
=o(1)
\end{align}
as $d\rightarrow\infty$. Thus we have
\begin{align*}
d^{1/2}\mathbb{E}_y[F_d\tilde{\alpha}_{y,d}(F_d)]\rightarrow 
\sqrt{1-\rho}(1+y(\log \tilde{q})'(y))\ (d\rightarrow\infty)
\end{align*}
uniformly in $y\in K_\delta$. 
Therefore by (\ref{f}), we have
\begin{equation}\nonumber
d\mathbb{E}_y[R_1^d-R_0^d]\rightarrow 2(1-\rho)y(1+y(\log \tilde{q})'(y))\ (d\rightarrow\infty)
\end{equation}
uniformly in $y\in K_\delta$ which completes the first part of the convergence of the triplet (\ref{triplet}). 
We prove the convergence of other two parts in (\ref{triplet}). 
By Lemma \ref{malliavinlemma},  
$\{F_d^2\}_d$ is $\mathbb{P}_y$-uniformly integrable in $d$ uniformly in  $y\in K_\delta$. By Proposition \ref{steinprop} together with this fact, we have
\begin{align*}
b_d(y)=d\mathbb{E}_y\left[\left(R_1^d-R_0^d\right)^2\right]
=4(1-\rho)y^2\mathbb{E}_y\left[F_d^2\alpha_{y,d}(F_d)\right]
\rightarrow 4(1-\rho)y^2\ (d\rightarrow\infty)
\end{align*}
uniformly in $y\in K_\delta$. In the same way, by uniform integrability of $\{F_d^4\}_d$, 
\begin{align*}
c_d(y)=d\mathbb{E}_y\left[\left(R_1^d-R_0^d\right)^4\right]
\le 4(1-\rho)y^2d^{-1}\mathbb{E}_y[|F_d|^4]=o(1)\ (d\rightarrow\infty). 
\end{align*}
Thus we obtain the uniform convergence of the triplet (\ref{triplet}) in $K_\delta$. If we prove the existence and uniqueness of the weak solution
of the stochastic differential equation (\ref{Proposition-1-eq1}), 
the convergence $Y^d\Rightarrow Y$ follows from Theorem IX.4.21 of \cite{JS}. 
 
 The existence and uniqueness comes from the standard approach. 
 Let $a(y)$ and $b(y)$ be as in (\ref{Proposition-1-eq1}). Let $c\in (0,\infty)$ and set the scale function $s(x)$
 so that 
 \begin{equation}\nonumber
 s(c)=0,\ s'(x)=\exp\left(-\int_c^x\frac{2a(u)}{b(u)^2}\dif u\right)=
 \exp\left(-\int_c^x\frac{2}{u}+(\log q)'(u)\dif u\right)=\frac{C}{x^2q(x)}
 \end{equation}
 for some constant $C>0$. 
 We use the convention such that $\int_a^b$ is $-\int_b^a$ if $b<a$. 
 By definition, $s(x)$ is a $C^2$ strictly increasing function. Now we prove 
 \begin{equation}\label{eq10}
 \lim_{x\rightarrow+\infty}s(x)=+\infty,\ 
  \lim_{x\rightarrow 0}s(x)=-\infty. 
 \end{equation}
By Schwarz's inequality, we have
 \begin{equation}\nonumber
\left|\int_c^x\frac{1}{u}\dif u\right|=
\left|\int_c^x \sqrt{q(u)}\frac{1}{\sqrt{u^2q(u)}}\dif u\right|\le \left|\int_c^x q(u)\dif u\right|^{1/2}\left|\int_c^x\frac{1}{u^2q(u)}\dif u\right|^{1/2}. 
 \end{equation}
 The left hand side tends to $+\infty$ as $x\rightarrow +\infty$ or 
  $x\rightarrow +0$ and the first term in the right-hand side 
 is bounded by $1$.  Therefore (\ref{eq10}) follows, 
and $s:(0,\infty)\rightarrow\mathbb{R}$ is a one-to-one map. 
 By It\^o's formula, $Z_t=s(Y_t)$ is the solution of the stochastic differential equation $\dif Z_t=\tilde{b}(Z_t)\dif W_t$ 
 where $\tilde{b}(x)=C/(\tilde{q}\circ s^{-1})(x)$ for some constant $C>0$ and for $\tilde{q}(x)=xq(x)$. 
 Thus it has the unique solution by Theorem 5.5.7 of \citet{Karatzas1991}. 
 By using the solution $Z_t$, we have the unique solution of (\ref{Proposition-1-eq1})
 by $Y_t=s^{-1}(Z_t)$. 
Hence $Y^d\Rightarrow Y$ follows by Theorem IX.4.21 of \cite{JS}.

Stationarity and ergodicity of $Y$ is yet to be proved. 
However stationarity comes from the fact that each $Y^d$ is stationary, and ergodicity comes from that of $\{Z_t\}_t$. 
Hence the claim follows. 
\end{proof}


\begin{proof}[Proof of Theorem \ref{Theorem-3}] 
By Proposition \ref{Proposition-1}, first we note that 
$\sup_{m\le M-1}|R_m^d-R_0^d|=o_{\mathbb{P}}(1)$ for any $M\in\mathbb{N}$. 
By this fact, observe that all proposed values of the MpCN algorithm are accepted for a finite number of iteration $M\in\mathbb{N}$ 
in probability $1$ since 
\begin{align*}
\mathbb{P}(X_{m-1}^d=X_m^d\ \exists m\in\{1,\ldots, M-1\})
&\le M\mathbb{P}(X_0^d=X_1^d)\\
&= M\left(1-\mathbb{E}[\alpha(X_0^d, X_1^{d*})]\right)\\
&=M\left(1-\mathbb{E}\left[\min\left\{1,\frac{\tilde{q}_d(R_1^{d*})}{\tilde{q}_d(R_0^d)}\right\}\right]\right)
\rightarrow 0
\end{align*}
by Lebesgue's dominated convergence theorem
and Lemma \ref{assumplem}. 
Thus $\left\{(R_m^d, S_m^d)\right\}_m$ defined in (\ref{cmaeq1}) converges weakly to
$\left\{(R_m, S_m)\right\}_m$ defined by
\begin{equation}\nonumber
\left\{\begin{array}{l}
R_m=R_0\\
S_m= \sqrt{\rho}S_{m-1}+\sqrt{1-\rho}(R_0)^{1/2}W_m,\ W_m\sim N_k(0,I_k)
\end{array}
\right. 
\end{equation}
for $m\ge 1$, 
where  $R_0\sim Q$ and $S_0\sim N_k(0, R_0I_k)$. 
By Proposition \ref{Proposition-1}, the process $Y^d=\left\{R_{[dt]}^d\right\}_t$ converges to a stationary ergodic process. 
Hence the claim follows by Lemma \ref{joint}. 
\end{proof}

\section*{Acknowledgement}
The author wishes to thank to Andreas Eberle, Ajay Jasra, Gareth O. Roberts and Masayuki Uchida for fruitful discussions.  
A part of this work was done when the author was visiting the  Institute for Applied Mathematics, Bonn University.  The author thanks the  Institute for Applied Mathematics, Bonn University for its hospitality. 

\appendix

\section{Some technical estimates}\label{appen1}
Set $K_\delta=[\delta,\delta^{-1}]$ for $\delta\in (0,1)$ and fix $\delta$ throughout. 

\subsection{Estimate  by using the Wiener chaos}

The following is a quick review of Malliavin calculus. 
For the detail, see monographs such as \citet{MR2200233} and \citet{MR2962301}. 

\begin{description}
\item[Abstract Wiener space]
Let $\mathfrak{H}$ be a separable Hilbert space with 
inner product $\left\langle\cdot,\cdot\right\rangle_\mathfrak{H}$ and the norm $\|h\|^2_\mathfrak{H}=\left\langle h,h\right\rangle_\mathfrak{H}$. 
Let $\left\{W(h);h\in \mathfrak{H}\right\}$ be an isonormal Gaussian process on $(\Omega,\mathcal{F},\mathbb{P})$, that is, $W(h)$ is centered Gaussian and 
$\mathbb{E}[W(g)W(h)] =\left\langle g,h\right\rangle_\mathfrak{H}$. 
The $\sigma$-algebra $\mathcal{F}$ is generated by $W$. 
This triplet $(W,\mathfrak{H},\mathbb{P})$
is called an abstract Wiener space. 
\item[Wiener-Chaos decomposition]
Let $L^2(\Omega)$ be the space of  square integrable random variables. 
Let $H_n(x)=(-1)^ne^{x^2/2}\frac{\dif^n}{\dif x^n}e^{-x^2/2}$ be the $n$-th Hermite polynomial. 
Write $\mathcal{H}_n$ for the linear subspace of $L^2(\Omega)$ generated by 
$\left\{H_n(W(h));h\in \mathfrak{H}\right\}$. The linear space $\mathcal{H}_n$ is called the $n$-th Wiener chaos.  Then any element $F\in L^2(\Omega)$ can be described by $F=\mathbb{E}[F]+\sum_{n=1}^\infty F_n$
for $F_n\in \mathcal{H}_n$, that is, $L^2(\Omega)=\bigoplus_{n=0}^\infty\mathcal{H}_n$, where $\mathcal{H}_0$ is the set of constants. This is called the Wiener-Chaos decomposition or the Wiener-It\^o decomposition. 
\item[Fr\'echet derivative]
A smooth random variables is a random variable with the form
$F=f(W(h_1),\ldots, W(h_n))$
where $h_i\in \mathfrak{H}$ and $f$ is a $C^\infty$ function such that all derivatives have polynomial growth.
Then Fr\'echet derivative of $F$ is defined by 
\begin{equation}\nonumber
DF=\sum_{i=1}^n\frac{\partial f}{\partial x_i}(W(h_1),\ldots, W(h_n))h_i
\end{equation}
and so $DF$ is a random variable with values in $\mathfrak{H}$. 
We set 
\begin{equation}\nonumber
\|F\|_{\mathbb{D}^{1,2}}:=\left(\mathbb{E}\left[|F|^2\right]+\mathbb{E}\left[\|DF\|^2_\mathfrak{H}\right]\right)^{1/2}. 
\end{equation}
Write $\mathbb{D}^{1,2}$ for the closure of the space of smooth random variables  with respect to the norm $\|\cdot\|_{\mathbb{D}^{1,2}}$
and extend $D$ to $\mathbb{D}^{1,2}$. 
\item[Ornstein-Uhlenbeck semigroup]
The Ornstein-Uhlenbeck semigroup $(P_t)_{t\ge 0}$
is defined by 
\begin{equation}\nonumber
P_tF=\mathbb{E}[F]+\sum_{n=1}^\infty e^{-nt}F_n
\end{equation}
for $F=\mathbb{E}[F]+\sum_{n=1}^\infty F_n\ (F_n\in \mathcal{H}_n)$. 
The operator $L$  and $L^{-1}$ is defined by
\begin{equation}\nonumber
LF=\sum_{n=1}^\infty -nF_n,\ L^{-1}F=\sum_{n=1}^\infty -F_n/n
\end{equation}
where $LF$ can be defined if $\sum n^2\mathbb{E}[|F_n|^2]<\infty$. 
\end{description}



By the so-called hypercontractivity property of Ornstein-Uhlenbeck operator, we have the following
for finite Wiener chaoses. 
See Corollary 2.8.14 of \citet{MR2962301} for the proof. 

\begin{proposition}\label{contractive}
Let $F\in\mathcal{H}_n$. Then for $p>2$, 
\begin{equation}\nonumber
\mathbb{E}[|F|^p]^{1/p}\le (p-1)^{n/2}\mathbb{E}[|F|^2]^{1/2}. 
\end{equation}
\end{proposition}

By using this, we prove the following bounds for the chi-squared distribution. 

\begin{lemma}\label{zlemma}
For $d\in\mathbb{N}$, $\xi_d$ follows the chi-squared distribution with $d$ degrees of freedom. 
Then 
\begin{equation}\nonumber
\sup_d d\mathbb{E}\left[\left|\left(\frac{\xi_d}{d}\right)^{k/2}-1\right|^2\right]<\infty\ (k\in\mathbb{Z})
\ 
\mathrm{and}\ 
\sup_d d^{k/2}\mathbb{E}\left[\left|\frac{\xi_d}{d}-1\right|^k\right]<\infty\ (k\in\mathbb{N}). 
\end{equation}
\end{lemma}

\begin{proof}
By definition, for $k\in\mathbb{Z}$, 
\begin{equation}\label{chisqeq}
0\le \mathbb{E}\left[\left(\frac{\xi_d}{d}\right)^k\right]-1=\left(\frac{2}{d}\right)^k\frac{\Gamma(\frac{d}{2}+k)}{\Gamma(\frac{d}{2})}-1
\le 
\left(1+\frac{2}{d}k\right)^{k}-1=O(d^{-1}). 
\end{equation}
Observe that
\begin{equation}\nonumber
\left\{\left(\frac{\xi_d}{d}\right)^{k/2}-1\right\}^2
\le \left\{\left(\frac{\xi_d}{d}\right)^{k/2}-\left(\frac{\xi_d}{d}\right)^{-k/2}\right\}^2=
\sum_{i=\pm k}\left(\frac{\xi_d}{d}\right)^i-1. 
\end{equation}
Hence the first claim follows from (\ref{chisqeq}). 
Observe that $\xi_d/d-1$ has the same law as $d^{-1}\sum_{i=1}^d\left(W(e_i)^2-1\right)$,  which is in the second Wiener chaos in $(W,\mathfrak{H},\mathbb{P})$. 
Then the  second claim comes from Proposition \ref{contractive} since we have
\begin{equation}\nonumber
\mathbb{E}\left[\left|d^{1/2}\left(\frac{\xi_d}{d}-1\right)\right|^k\right]^{1/k}\le (k-1)
\mathbb{E}\left[\left|d^{1/2}\left(\frac{\xi_d}{d}-1\right)\right|^2\right]^{1/2}=\sqrt{2}(k-1).
\end{equation}
\end{proof}

The following is the key result for our paper. See Theorem 2.9.1 \citet{MR2962301}
for the proof (see also the proof of Theorem 3.1 of \citet{MR2520122}). 

\begin{proposition}\label{malliavinprop}
For $F_d\in\mathbb{D}^{1,2}$, suppose that  
\begin{equation}\label{malliavinpropeq1}
\mathbb{E}\left|\left\langle DF_d,-DL^{-1}F_d\right\rangle_\mathfrak{H} -1\right|= O(d^{-1/2})
\end{equation}
and $F_d$ has a density with respect to the Lebesgue measure. 
Then for any absolutely continuous function $f$, 
\begin{equation}\nonumber
\mathbb{E}\left[\left(F_d-\mathbb{E}\left[F_d\right]\right)f(F_d)\right]-\mathbb{E}\left[f'(F_d)\right]=O\left(d^{-1/2}\|f'\|_\infty\right). 
\end{equation}
\end{proposition}


\subsection{Representation of random variables  for the MpCN algorithm}\label{localmalliavin}
We introduce an abstract Wiener space  to the MpCN algorithm.  
Write $\mathbb{P}_y$ and $\mathbb{E}_y$ for 
the conditional probability and the expectation with respect to $\mathbb{P}$ given $y=\|X_0^d\|^2/d$ with respectively. 
Assume that the orthonormal base of $\mathfrak{H}$ is $\left\{e_i;i\in\mathbb{Z}\right\}$
and consider an abstract Wiener space $(W,\mathfrak{H},\mathbb{P}_y)$ for each $y\in (0,\infty)$. 
Set 
\begin{equation}\nonumber
 \|F\|_{\mathbb{D}^{1,2}_\delta}=\sup_{y\in K_\delta}\left(\mathbb{E}_y\left[|F|^2\right]+\mathbb{E}_y\left[\|DF\|_\mathfrak{H}^2\right]\right)^{1/2}.
 \end{equation} 
 Rewrite random variables defined in (\ref{MpCN}) for $m=1$ as random variables in  $(W,\mathfrak{H},\mathbb{P}_y)$ by
\begin{subequations}\label{subeq}
\begin{align}
&X_1^{d*}=\sqrt{\rho}x+\sqrt{(1-\rho)Z_1^d}W_1^d,\\ 
&W_1^d=\sum_{i=1}^dW(e_i)e_i,\\ 
&Z_1^d=\frac{yd}{\|\tilde{W}_1^d\|^2_\mathfrak{H}}=\frac{\|x\|_\mathfrak{H}^2}{\|\tilde{W}_1^d\|^2_\mathfrak{H}}, 
\end{align}
\end{subequations}
where $\tilde{W}_1^d=\sum_{i=1}^dW(e_{-i})e_{-i}$
and $x=\sum_{i=1}^dx_ie_i$ is any value such that $\|x\|^2_\mathfrak{H}/d=y$. Notice that this representation does not change the law of $F_d$ which is defined in (\ref{Fd}) and that defined here:
\begin{equation}\label{c}
F_d=d^{1/2}\frac{\|X_1^{d*}\|_\mathfrak{H}^2-\|x\|_\mathfrak{H}^2}{2\|x\|_\mathfrak{H}^2\sqrt{1-\rho}}. 
\end{equation}

\begin{lemma}\label{malliavinlemma}
Let $G_d^k=d^{1/2}\left(\left(\frac{\|\tilde{W}_1^d\|_\mathfrak{H}^2}{d}\right)^{k/2}-1\right)\ (k\in\mathbb{Z})$.
Then 
for each $\delta\in (0,1)$, $k\in\mathbb{Z}$, 
\begin{equation}\nonumber
\sup_d\|G_d^k\|_{\mathbb{D}_\delta^{1,2}}<\infty,\ \sup_d\|F_d\|_{\mathbb{D}_\delta^{1,2}}<\infty. 
\end{equation}
Also we have
\begin{equation}\nonumber
\sup_d\sup_{y\in K_\delta}\mathbb{E}_y[|F_d|^4]<\infty. 
\end{equation}
\end{lemma}

\begin{proof}
Note that the law of $G_d^k$ and $F_d$ do not depend on $y$ and so we omit the subscript $y$ in this proof. 
First we prove $\sup_d\|G_d^k\|_{\mathbb{D}_\delta^{1,2}}<\infty$. The $L^2$ boundedness
$\sup_d\mathbb{E}[|G_d^k|^2]<\infty$ was proved in 
Lemma \ref{zlemma}. Also 
\begin{align*}\nonumber
\|DG_d^k\|_\mathfrak{H}
&=\left\|k\left(\frac{\|\tilde{W}_1^d\|_\mathfrak{H}^2}{d}\right)^{k/2-1}\frac{\tilde{W}_1^d}{d^{1/2}}\right\|_{\mathfrak{H}}
=k\left(\frac{\|\tilde{W}_1^d\|_\mathfrak{H}^2}{d}\right)^{(k-1)/2}
\end{align*}
and hence $\sup_d\mathbb{E}[\|DG_d^k\|_\mathfrak{H}^2]<\infty$ follows by Lemma \ref{zlemma}. 
and hence the first claim follows. 

Next we show $\sup_d\|F_d\|_{\mathbb{D}_\delta^{1,2}}<\infty$. By (\ref{crgeq2}), 
\begin{equation}\label{b}
F_d=\sqrt{\rho}\left(\frac{\|\tilde{W}_1^d\|^2_\mathfrak{H}}{d}\right)^{-1/2}\left\langle \frac{x}{\|x\|_\mathfrak{H}},W_1^d\right\rangle_\mathfrak{H}
+\sqrt{1-\rho}\left(\frac{\|\tilde{W}_1^d\|^2_\mathfrak{H}}{d}\right)^{-1}\left(d^{-1/2}\frac{\|W_1^d\|^2_\mathfrak{H}-\|\tilde{W}_1^d\|^2_\mathfrak{H}}{2}\right). 
\end{equation}
It is not difficult to check
$H_1^d:=\left\langle \frac{x}{\|x\|_\mathfrak{H}},W_1^d\right\rangle_\mathfrak{H}\in \mathcal{H}_1$
and $H_2^d:=d^{-1/2}\frac{\|W_1^d\|^2_\mathfrak{H}-\|\tilde{W}_1^d\|^2_\mathfrak{H}}{2}\in\mathcal{H}_2$
satisfy $\sup_d\|H_i^d\|_{\mathbb{D}^{1,2}_\delta}<\infty\ (i=1,2)$. 
This, together with the first claim prove
$\sup_d\|F_d\|_{\mathbb{D}_\delta^{1,2}}<\infty$
by H\"older's inequality and Minkowski's inequality. 

Finally we check $\sup_d\mathbb{E}[|F_d|^4]<\infty$. 
However by Proposition \ref{malliavinprop}, 
$\sup_d\mathbb{E}[|H_i^d|^4]^{1/4}\le 3^{i/2}\sup_d\mathbb{E}[|H_i^d|^2]^{1/2}=3^{i/2}<\infty$. 
Hence it is sufficient to show $\sup_d\mathbb{E}\left[\left(\frac{\|\tilde{W}_1^d\|^2_\mathfrak{H}}{d}\right)^{-4}\right]<\infty$
  by H\"older's inequality and Minkowski's inequality. 
However this comes from Lemma \ref{zlemma} and hence the claim follows. 
\end{proof}

\begin{proposition}\label{main}
Suppose that $f$ is an absolutely continuous function. 
Then for $\delta\in (0,1)$, 
\begin{equation}\nonumber
\sup_{y\in K_\delta}\left|\mathbb{E}_y[F_df(F_d)]-\mathbb{E}_y[f'(F_d)]-d^{-1/2}\sqrt{1-\rho}\mathbb{E}_y[f(F_d)]\right|=O\left(\max\left\{d^{-1/2}\|f'\|_\infty,d^{-1}\|f\|_\infty\right\}\right)
\end{equation}
\end{proposition}

\begin{proof}
We check the conditions in Proposition \ref{malliavinprop}. 
Without loss of generality, we can certainly assume that $\|f'\|_\infty<\infty$
and $\|f\|_\infty<\infty$ 
since otherwise the right hand side becomes $+\infty$. 
We have a decomposition of $F_d$ as in (\ref{b}). 
Set 
\begin{equation}\nonumber
F_{d,0}
=\sqrt{\rho}\left\langle \frac{x}{\|x\|_\mathfrak{H}}, W^d_1\right\rangle_\mathfrak{H}+\sqrt{1-\rho}\left(d^{-1/2}\frac{\|W_1^d\|_\mathfrak{H}^2-\|\tilde{W}_1^d\|_\mathfrak{H}^2}{2}\right).
\end{equation}
By Lemma \ref{malliavinlemma} together with H\"older's inequality and Minkowski's inequality, we have
\begin{equation}\label{maineq1}
d^{1/2}\|F_d-F_{d,0}\|_{\mathbb{D}^{1,2}_\delta}<\infty.
\end{equation} 
By simple algebra, 
\begin{align*}
\left\langle DF_{d,0},-DL^{-1}F_{d,0}\right\rangle_\mathfrak{H} &=
\left\langle \sqrt{\rho}\frac{x}{\|x\|_\mathfrak{H}}+\sqrt{1-\rho}d^{-/2}(W_1^d-\tilde{W}_1^d),\sqrt{\rho}\frac{x}{\|x\|_\mathfrak{H}}+\sqrt{1-\rho}d^{-1/2}\frac{W_1^d-\tilde{W}_1^d}{2}\right\rangle_\mathfrak{H}\\
&=1+\frac{3}{2}\sqrt{\rho(1-\rho)}d^{-1/2}\left\langle \frac{x}{\|x\|_\mathfrak{H}},W_1^d\right\rangle_\mathfrak{H}
+(1-\rho)d^{-1}\frac{(\|W_1^d\|_\mathfrak{H}^2-d)+(\|\tilde{W}_1^d\|_\mathfrak{H}^2-d)}{2}
\end{align*}
and so it is straightforward to check $\mathbb{E}_y\left|\left\langle DF_{d,0},-DL^{-1}F_{d,0}\right\rangle_\mathfrak{H}-1\right|=O(d^{-1/2})$ uniformly in $y$. 
Therefore (\ref{malliavinpropeq1}) follows from (\ref{maineq1})
by H\"older's and Mikowskii's inequalities together with 
$\mathbb{E}_y[\|DL^{-1}F\|_\mathfrak{H}^2]\le \mathbb{E}_y[\|DF\|_\mathfrak{H}^2]$ for $F\in\mathbb{D}^{1,2}_\delta$. 
Also, since $F_d$ is a mixture of finite multiple Wiener chaoses, it has a density with respect to Lebesgue measure by Theorem 5.1 of \citet{MR582167}. 
 Thus we can apply Proposition \ref{malliavinprop}. 
In the current case, 
\begin{equation}\nonumber
d^{1/2}\mathbb{E}_y[F_d]=\frac{\sqrt{1-\rho}}{2}d\mathbb{E}_y\left[\frac{d}{\|\tilde{W}_1^d\|^2_\mathfrak{H}}-1\right]=\sqrt{1-\rho}+O(d^{-1})
\end{equation}
since the mean of the inverse chi-squared distribution $d/\|\tilde{W}_1^d\|^2_\mathfrak{H}$ is $1/(d-2)$. 
\end{proof}

\subsection{Total variation distance and Stein's method}

Total variation distance of measures $\mu$ and $\nu$ on a measurable space $(E,\mathcal{E})$ is defined by
\begin{equation}\nonumber
\|\mu-\nu\|_{\mathrm{TV}}=\sup_{A\in\mathcal{E}}|\mu(A)-\nu(A)|=\frac{1}{2}\sup\left|\int_Ef(x)\mu(\dif x)-\int_Ef(x)\nu(\dif x)\right|
\end{equation}
where the second supremum  is taken for all $[-1,1]$-valued measurable function on $(E,\mathcal{E})$. 
The convergence in total variation distance is stronger than weak convergence. However for sequences from 
finite Wiener chaoses, \citet{MR3003367} obtain the following useful result. 

\begin{theorem}[Theorem 5.1 of \citet{MR3003367}]\label{NourdinPoly}
Let $F_n=(F_{n,1},\ldots, F_{n,k})$ be a random vector such that $F_{n,i}\in\mathcal{H}_{h_i}$ for $h_1,\ldots, h_k\in\mathbb{N}$. 
If $\mathcal{L}(F_n)$ converges weakly to $N_k(0, C)$ with $\det C>0$, then the total variation convergence also holds. 
\end{theorem}

Stein's method is an efficient tool to estimate the total variation distance of probability measures in $\mathbb{R}$. 
See \citet{Stein} for general reference and see also \citet{MR2962301} for beautiful relation to Malliavin calculus. 
A fundamental result is that for any measurable function $h$ such that $\|h\|_\infty\le 1$, there 
is a function $f$ called Stein's solution such that
\begin{equation}\nonumber
h(x)-Nh=f'(x)-xf(x).
\end{equation}
Moreover, the solution is absolutely continuous and $\|f\|_\infty\le \sqrt{\pi/2}$ and $\|f'\|_\infty\le 2$ (see Lemma 2.4 of \citet{Stein}). 
Immediate corollary of this fact is that 
\begin{equation}\nonumber
\|\mathcal{L}(X)-N\|_{\mathrm{TV}}\le \sup_{f\in\mathfrak{F}}\left|\mathbb{E}\left[f'(X)\right]-\mathbb{E}\left[Xf(X)\right]\right|
\end{equation}
where $\mathfrak{F}$ is a set of functions such that $\|f\|_\infty\le \sqrt{\pi/2}$ and $\|f'\|_\infty\le 2$. 
In particular, we have the following. 

\begin{proposition}\label{steinprop}
For $\delta>0$ and the random variable $F_d$ defined in (\ref{c}), 
$\sup_{y\in K_\delta}\|\mathcal{L}_{\mathbb{P}_y}(F_d)-N\|_{\mathrm{TV}}\rightarrow 0$. 
\end{proposition}
Stein's method was a basic tool in  \citet{arXiv:1406.5392} and implicitly used throughout in this paper. 


\section{Elements of consistency of MCMC}\label{appen2}

\subsection{Some sufficient conditions for consistency}

The following lemma is a fundamental result for consistency of MCMC. 

\begin{lemma}[Lemma 2 of \citet{Kamatani10}]\label{lem1}
Let $\xi^d=\{\xi_m^d\}_m$ be a sequence of stationary process on $\mathbb{R}^k$. 
If $\xi^d$ converges in law to $\xi=\{\xi_m\}_m$, and if $\xi$ is a stationary ergodic process, then 
the law of $\xi^d$ is consistent in the sense of (\ref{consistency}). 
\end{lemma}

We need a slightly generalization of this lemma. 
Let $k_1, k_2\in\mathbb{N}$. 
Suppose that $\mathbb{R}^{k_1+k_2}$-valued random variable $X_m^d$ has two parts, $X_m^d=(X_m^{d,1}, X_m^{d,2})$ where $X_m^{d,i}$ is $\mathbb{R}^{k_i}$ valued for each $i=1,2$. 
Corresponding to $X^{d,1}$ and $X^{d,2}$, 
the invariant probability measure has the following decomposition	
\begin{equation}\nonumber
P_d(\dif x_1 \dif x_2)=P_d^1(\dif x_1)P_d^2(\dif x_2|x_1). 
\end{equation}
Furthermore, we assume the following. Let $T_d\rightarrow\infty$. 
Let $[x]$ be the integer part of $x\ge 0$. 

\begin{assumption}
\begin{enumerate}
\item 
For $Y_t^{d,1}=X^{d,1}_{[T_dt]}$, $Y^{d,1}\Rightarrow  Y^1$ (in Skorohod's sense) where $Y^1$ is stationary and ergodic
continuous process 
with the invariant probability measure $P^1$. 
\item Random variables $X^d=\{X^d_m\}_m$ converges to $X=\{X_m\}_m=\{(X_0^1, X^2_m(X_0^1))\}_m$
where $X^2(x)=\{X^2_m(x)\}_m$ is a stationary and ergodic process 
with the invariant probability measure $P^{2|1}(\cdot|x)$ for each $x$, and $X_0^1\sim P^1$. 
\item For any bounded continuous function $f$, $P^{2|1}f(x_1)=\int f(x_1,x_2)P^{2|1}(\dif x_2|x_1)$ is continuous in $x_1$. 
\end{enumerate}
\end{assumption}

The proof of the following lemma is essentially same as 
Lemma B.2 of \citet{arXiv:1406.5392}. Thus we omit it. 

\begin{lemma}\label{joint}
Under the above assumption, 
\begin{equation}\nonumber
\frac{1}{M_d}\sum_{m=0}^{M_d-1}f(X_m^d)-P_d(f)=o_\mathbb{P}(1)
\end{equation}
for any continuous and bounded function $f$ and for $M_d\rightarrow\infty$ 
such that $M_d/T_d\rightarrow\infty$. 
\end{lemma}

\subsection{Consistency of the Metropolis-Hastings algorithm}
We prove  consistency of the  Metropolis-Hastings (MH)
algorithm. 
For probability measures $P, Q$ and a transition kernel $K$ on $(E,\mathcal{E})$, we introduce an operator $\otimes$ and $T$ for any set $A\times B=\left\{(x,y);x\in A, y\in B\right\}$ by 
\begin{equation}\nonumber
(P\otimes K)(A\times B)=\int_A P(\dif x)K(x,B),\ 
(P\otimes K)^T(A\times B)=(P\otimes K)(B\times A)
\end{equation}
and extend them to probability measures on $\mathcal{E}^{\otimes 2}$ by Hahn-Kolmogorov's theorem. 
We introduce  another operator $\wedge$ by
\begin{equation}\nonumber
(P\wedge Q)(\dif x)=\min\left\{p(x),q(x)\right\}\sigma(\dif x). 
\end{equation}
where $p(x)$ and $q(x)$ are the Radon-Nikod\'ym derivatives of $P$ and $Q$ with respect to a $\sigma$-finite measure $\sigma(\dif x)$. 
Let $X^n=\left\{X^n_m;m\in\mathbb{N}_0\right\}$
be a stationary Markov chain  with the transition kernel $K_n$ with the initial distribution $P_n$, 
and let $X=\left\{X_m;m\in\mathbb{N}_0\right\}$ be that for 
the transition kernel $K$ with the initial distribution $P$.

\begin{lemma}[Lemmas 2 and 3 of \citet{Kamatani10}]\label{MHlemma1}
Let  $K$ and $K_n\ (n=1,2,\ldots)$
be transition kernels that have the invariant probability distributions 
$P$ and $P_n$ with respectively. 
If $\|P_n\otimes K_n-P\otimes K\|_{\mathrm{TV}}\rightarrow 0$, then 
$X^n$ tends to $X$ in law. 
\end{lemma}

Thus $\|P_n\otimes K_n-P\otimes K\|_{\mathrm{TV}}\rightarrow 0$
with ergodicity of $K$ is a set of sufficient conditions for consistency. 
The transition kernel $K$ of the Metropolis-Hastings algorithm 
with the proposal transition kernel $Q(x,\dif y)=q(x,y)\sigma(\dif y)$ ($q$ is supposed to be $\mathcal{E}^{\otimes 2}$-measurable) is 
\begin{equation}\nonumber
K(x,\dif y)=Q(x,\dif y)\min\left\{1,\frac{p(y)q(y,x)}{p(x)q(x,y)}\right\}
+R(x)\delta_x(\dif y)
\end{equation}
where 
\begin{equation}\label{reject}
R(x)=1-\int_{y\in E}Q(x,\dif y)\min\left\{1,\frac{p(y)q(y,x)}{p(x)q(x,y)}\right\}. 
\end{equation}
Thus 
\begin{equation}\nonumber
(P\otimes K)(\dif x,\dif y)=(P\otimes Q)\wedge(P\otimes Q)^T(\dif x,\dif y)+PR(\dif x)\delta_x(\dif y)
\end{equation}
where
\begin{equation}\label{pr}
PR(\dif x):=P(\dif x)R(x)=P(\dif x)-(P\otimes Q)\wedge (P\otimes Q)^T(\dif x\times E). 
\end{equation}

The following lemma shows that the total variation convergence of the transition kernel of the Metropolis-Hastings algorithm comes from that of the proposal transition kernel. 

\begin{lemma}\label{MHlemma2}
Suppose $K_1$ and $K_2$ are transition kernels of 
the Metropolis-Hastings algorithm  with the proposal transition kernels $Q_1$ and $Q_2$
and the target probability distribution $P_1$ and $P_2$
with respectively.  Then  
\begin{equation}\nonumber
\|P_1\otimes K_1-P_2\otimes K_2\|_{\mathrm{TV}}\le 6\|P_1\otimes Q_1-P_2\otimes Q_2\|_{\mathrm{TV}}. 
\end{equation}
\end{lemma}

\begin{proof}
By triangular inequality, 
\begin{align*}
\|P_1\otimes K_1-P_2\otimes K_2\|_{\mathrm{TV}}\le \|(P_1\otimes Q_1)\wedge (P_1\otimes Q_1)^T-(P_2\otimes Q_2)\wedge (P_2\otimes Q_2)^T\|_{\mathrm{TV}}
+\|P_1R_1-P_2R_2\|_{\mathrm{TV}}
\end{align*}
where $R_i(x)$ is the rejection probability defined in (\ref{reject}) of the transition kernel $K_i$ for  $i=1,2$. 
By (\ref{pr}), 
the second term in the right-hand side of the above is dominated by  twice of the first term. 
To find a bound of the first term, observe that for any $x_1,x_2,y_1,y_2\in\mathbb{R}$ we have
$|x_1\wedge x_2-y_1\wedge y_2|\le \sum_{i=1}^2|x_i-y_i|$ where $x\wedge y=\min\{x,y\}$. 
By this inequality, $\|\mu_1\wedge\mu_2-\nu_1\wedge\nu_2\|_{\mathrm{TV}}\le \sum_{i=1}^2\|\mu_i-\nu_i\|_{\mathrm{TV}}$. Thus we have 
\begin{align*}
\|P_1\otimes K_1-P_2\otimes K_2\|_{\mathrm{TV}}&\le 
3\|(P_1\otimes Q_1)\wedge (P_1\otimes Q_1)^T-(P_2\otimes Q_2)\wedge (P_2\otimes Q_2)^T\|_{\mathrm{TV}}\\
&\le 6\|P_1\otimes Q_1-P_2\otimes Q_2\|_{\mathrm{TV}}.
\end{align*}
\end{proof}

\end{document}